\documentclass[10pt,showpacs,letterpaper,aps,nofootinbib,prx,twocolumn,superscriptaddress]{revtex4-2} 
\usepackage{amsmath,amsfonts,amssymb,amsthm,bbm}
\usepackage[dvipsnames]{xcolor}
\usepackage{relsize}
\usepackage{comment}
\usepackage{stmaryrd}
\usepackage{changepage} 
\usepackage{graphicx}
\usepackage{enumitem}
\usepackage[colorlinks=true,linkcolor=blue,citecolor=lighterblue,urlcolor=lighterblue]{hyperref}
\usepackage{xcolor}
\definecolor{lighterblue}{rgb}{0.1,0.1,0.8}
\usepackage{enumitem}

\usepackage{mathtools}

\usepackage[colorlinks=true]{hyperref} 

\usepackage{dsfont}
\usepackage{yfonts}
\usepackage{rotating}
\usepackage{floatpag}
\rotfloatpagestyle{empty}

\usepackage{hyperref}

\usepackage{amsmath,amsthm,amssymb}
\usepackage{xspace,enumerate,color,epsfig}
\usepackage{graphicx}
\usepackage{marvosym}

\usepackage{tikzfig}
\usepackage{docmute}
\usepackage{keycommand}

\usepackage{pifont}

\usepackage{enumitem}

\newkeycommand{\pointmap}[style=point][1]{\,\tikz{\node[style=\commandkey{style}] (x) at (0,-0.3) {$#1$};\node [style=none] (3) at (0, 1.0) {};\node [style=none] (3) at (0, -1.0) {};\draw (x)--(0,0.7);}\,}

\newkeycommand{\typedkpoint}[style=kpoint,edgestyle=][2]{%
\,\begin{tikzpicture}
  \begin{pgfonlayer}{nodelayer}
    \node [style=none] (0) at (0, 1) {};
    \node [style=\commandkey{style}] (1) at (0, -0.75) {$#2$};
    \node [style=right label] (2) at (0.25, 0.5) {$#1$};
  \end{pgfonlayer}
  \begin{pgfonlayer}{edgelayer}
    \draw [\commandkey{edgestyle}] (1) to (0.center);
  \end{pgfonlayer}
\end{tikzpicture}\,}

\newkeycommand{\typedkpointdag}[style=kpoint adjoint,edgestyle=][2]{%
\,\begin{tikzpicture}
  \begin{pgfonlayer}{nodelayer}
    \node [style=none] (0) at (0, -1) {};
    \node [style=\commandkey{style}] (1) at (0, 0.75) {$#2$};
    \node [style=right label] (2) at (0.25, -0.5) {$#1$};
  \end{pgfonlayer}
  \begin{pgfonlayer}{edgelayer}
    \draw [\commandkey{edgestyle}] (0.center) to (1);
  \end{pgfonlayer}
\end{tikzpicture}\,}

\newcommand{\bistateadj}[1]{\,\begin{tikzpicture}[yshift=-3mm]
  \begin{pgfonlayer}{nodelayer}
    \node [style=kpointadj, minimum width=1 cm, inner sep=2pt] (0) at (0, 0.5) {$\psi$};
    \node [style=none] (1) at (-0.75, 0.25) {};
    \node [style=none] (2) at (0.75, 0.25) {};
    \node [style=none] (3) at (-0.75, -0.5) {};
    \node [style=none] (4) at (0.75, -0.5) {};
  \end{pgfonlayer}
  \begin{pgfonlayer}{edgelayer}
    \draw (1.center) to (3.center);
    \draw (2.center) to (4.center);
  \end{pgfonlayer}
\end{tikzpicture}\,}

\newkeycommand{\pointketbra}[style1=copoint,style2=point][2]{\,%
\begin{tikzpicture}
  \begin{pgfonlayer}{nodelayer}
    \node [style=none] (0) at (0, -1.5) {};
    \node [style=\commandkey{style1}] (1) at (0, -0.75) {$#1$};
    \node [style=\commandkey{style2}] (2) at (0, 0.75) {$#2$};
    \node [style=none] (3) at (0, 1.5) {};
  \end{pgfonlayer}
  \begin{pgfonlayer}{edgelayer}
    \draw (0.center) to (1);
    \draw (2) to (3.center);
  \end{pgfonlayer}
\end{tikzpicture}\,}

\newkeycommand{\twocopointketbra}[style1=copoint,style2=copoint,style3=point][3]{\,%
\begin{tikzpicture}
  \begin{pgfonlayer}{nodelayer}
    \node [style=none] (0) at (-0.75, -1.5) {};
    \node [style=none] (1) at (0.75, -1.5) {};
    \node [style=\commandkey{style1}] (2) at (-0.75, -0.75) {$#1$};
    \node [style=\commandkey{style2}] (3) at (0.75, -0.75) {$#2$};
    \node [style=\commandkey{style3}] (4) at (0, 0.75) {$#3$};
    \node [style=none] (5) at (0, 1.5) {};
  \end{pgfonlayer}
  \begin{pgfonlayer}{edgelayer}
    \draw (0.center) to (2);
    \draw (1.center) to (3);
    \draw (4) to (5.center);
  \end{pgfonlayer}
\end{tikzpicture}\,}

\newkeycommand{\twopointketbra}[style1=copoint,style2=point,style3=point][3]{\,%
\begin{tikzpicture}
  \begin{pgfonlayer}{nodelayer}
    \node [style=none] (0) at (0, -1.5) {};
    \node [style=none] (1) at (-0.75, 1.5) {};
    \node [style=\commandkey{style1}] (2) at (0, -0.75) {$#1$};
    \node [style=\commandkey{style2}] (3) at (-0.75, 0.75) {$#2$};
    \node [style=\commandkey{style3}] (4) at (0.75, 0.75) {$#3$};
    \node [style=none] (5) at (0.75, 1.5) {};
  \end{pgfonlayer}
  \begin{pgfonlayer}{edgelayer}
    \draw (0.center) to (2);
    \draw (1.center) to (3);
    \draw (4) to (5.center);
  \end{pgfonlayer}
\end{tikzpicture}\,}

\newkeycommand{\pointbraket}[style1=point,style2=copoint][2]{\,%
\begin{tikzpicture}
  \begin{pgfonlayer}{nodelayer}
    \node [style=\commandkey{style1}] (0) at (0, -0.5) {$#1$};
    \node [style=\commandkey{style2}] (1) at (0, 0.5) {$#2$};
  \end{pgfonlayer}
  \begin{pgfonlayer}{edgelayer}
    \draw (0) to (1);
  \end{pgfonlayer}
\end{tikzpicture}\,}

\newkeycommand{\kpointbraket}[style1=kpoint,style2=kpoint adjoint][2]{\,%
\begin{tikzpicture}
  \begin{pgfonlayer}{nodelayer}
    \node [style=\commandkey{style1}] (0) at (0, -0.75) {$#1$};
    \node [style=\commandkey{style2}] (1) at (0, 0.75) {$#2$};
  \end{pgfonlayer}
  \begin{pgfonlayer}{edgelayer}
    \draw (0) to (1);
  \end{pgfonlayer}
\end{tikzpicture}\,}

\newkeycommand{\kpointmap}[style1=kpoint,style2=map][2]{\,%
\begin{tikzpicture}
  \begin{pgfonlayer}{nodelayer}
    \node [style=\commandkey{style1}] (0) at (0, -1.1) {$#1$};
    \node [style=\commandkey{style2}] (1) at (0, 0.2) {$#2$};
    \node [style=none] (2) at (0, 1.5) {};
  \end{pgfonlayer}
  \begin{pgfonlayer}{edgelayer}
    \draw (0) to (1);
    \draw (1) to (2.center);
  \end{pgfonlayer}
\end{tikzpicture}%
\,}

\newkeycommand{\sandwichtwo}[style1=point,style2=map,style3=map,style4=copoint][4]{\,%
\begin{tikzpicture}
  \begin{pgfonlayer}{nodelayer}
    \node [style=\commandkey{style2}] (0) at (0, -0.75) {$#2$};
    \node [style=\commandkey{style3}] (1) at (0, 0.75) {$#3$};
    \node [style=\commandkey{style1}] (2) at (0, -2) {$#1$};
    \node [style=\commandkey{style4}] (3) at (0, 2) {$#4$};
  \end{pgfonlayer}
  \begin{pgfonlayer}{edgelayer}
    \draw (2) to (0);
    \draw (0) to (1);
    \draw (1) to (3);
  \end{pgfonlayer}
\end{tikzpicture}\,}

\newkeycommand{\longbraket}[style1=point,style2=copoint][2]{\,%
\begin{tikzpicture}
  \begin{pgfonlayer}{nodelayer}
    \node [style=\commandkey{style1}] (0) at (0, -2) {$#1$};
    \node [style=\commandkey{style2}] (1) at (0, 2) {$#2$};
  \end{pgfonlayer}
  \begin{pgfonlayer}{edgelayer}
    \draw (0) to (1);
  \end{pgfonlayer}
\end{tikzpicture}\,}

\newkeycommand{\onb}[style=point]{\ensuremath{\left\{\tikz{\node[style=\commandkey{style}] (x) at (0,-0.3) {$j$};\draw (x)--(0,0.7);}\right\}}\xspace}

\newkeycommand{\copointmap}[style=copoint][1]{\,\tikz{\node[style=\commandkey{style}] (x) at (0,0.3) {$#1$};\draw (x)--(0,-0.7);}\,}

\tikzstyle{cdiag}=[matrix of math nodes, row sep=3em, column sep=3em, text height=1.5ex, text depth=0.25ex,inner sep=0.5em]
\tikzstyle{arrow above}=[transform canvas={yshift=0.5ex}]
\tikzstyle{arrow below}=[transform canvas={yshift=-0.5ex}]

\usetikzlibrary{shapes.multipart}

\tikzstyle{xor}=[circuit ee IEC, bulb,fill=white,rotate=45]
\tikzstyle{merge}=[circuit ee IEC, bulb,fill=white]
\tikzstyle{dmerge}=[circuit ee IEC, bulb,fill=black!20]
\tikzset{->-/.style={decoration={
  markings,
  mark=at position .5 with {\arrow{>}}},postaction={decorate}}}
\tikzset{-<-/.style={decoration={
  markings,
  mark=at position .5 with {\arrow{<}}},postaction={decorate}}}

\tikzstyle{bwSpider}=[
       rectangle split,
       rectangle split parts=2,
       rectangle split part fill={black,white},
 minimum size=3.6 mm, inner sep=-2mm, draw=black,scale=0.5,rounded corners=0.8 mm
       ]
 \tikzstyle{wbSpider}=[
       rectangle split,
       rectangle split parts=2,
       rectangle split part fill={white,black},
 minimum size=3.6 mm, inner sep=-2mm, draw=black,scale=0.5,rounded corners=0.8 mm
       ]

\tikzstyle{epiCopoint}=[regular polygon,regular polygon sides=3,draw,scale=0.75,inner sep=-0.5pt,minimum width=5mm,fill=white,regular polygon rotate=0,line width=1pt]
\tikzstyle{epiPoint}=[regular polygon,regular polygon sides=3,draw,scale=0.75,inner sep=-0.5pt,minimum width=5mm,fill=white,regular polygon rotate=180,line width=1pt]
\tikzstyle{epiPointWide}=[regular polygon,regular polygon sides=3,draw,scale=0.75,inner sep=-0.5pt,minimum width=8mm,fill=white,regular polygon rotate=180,line width=1pt]
\tikzstyle{epiBox}=[fill=white,draw, line width = 1pt,inner sep=0.6mm,font=\footnotesize,minimum height=3mm,minimum width=3mm]
\tikzstyle{epiBoxWide}=[fill=white,draw, line width = 1pt,inner sep=0.6mm,font=\footnotesize,minimum height=3mm,minimum width=5mm]
\tikzstyle{epiBoxVeryWide}=[fill=white,draw, line width = 1pt,inner sep=0.6mm,font=\footnotesize,minimum height=3mm,minimum width=7mm]
\tikzstyle{bWire}=[line width = .5pt, color=black]
\tikzstyle{oWire}=[line width = .5pt, color=black!50, double, double distance = 1pt]
\tikzstyle{osWire}=[line width = .5pt, color=black!50]
\tikzstyle{qWire}=[line width = 1pt, color=black]
\tikzstyle{cWire}=[color=gray,line width = .75pt]
\tikzstyle{CqWire}=[color=gray,line width = .75pt,->-]
\tikzstyle{CcWire}=[color=gray,line width = .75pt,->-]
\tikzstyle{RqWire}=[line width = 1pt, color=black,-<-]
\tikzstyle{RcWire}=[color=gray,line width = .75pt,-<-]
\tikzstyle{env}=[copoint,regular polygon rotate=0,minimum width=0.2cm, fill=black]

\tikzstyle{probs}=[shape=semicircle,fill=white,draw=black,shape border rotate=180,minimum width=1.2cm]

\tikzstyle{every picture}=[baseline=-0.25em,scale=0.5]
\tikzstyle{dotpic}=[] 
\tikzstyle{diredges}=[every to/.style={diredge}]
\tikzstyle{math matrix}=[matrix of math nodes,left delimiter=(,right delimiter=),inner sep=2pt,column sep=1em,row sep=0.5em,nodes={inner sep=0pt},text height=1.5ex, text depth=0.25ex]

\tikzstyle{inline text}=[text height=1.5ex, text depth=0.25ex,yshift=0.5mm]
\tikzstyle{label}=[font=\footnotesize,text height=1.5ex, text depth=0.25ex,yshift=0.5mm]
\tikzstyle{left label}=[label,anchor=east,xshift=1.5mm]
\tikzstyle{right label}=[label,anchor=west,xshift=-1mm]
\tikzstyle{up label}=[label,anchor=south,yshift=-1mm]

\tikzstyle{braceedge}=[decorate,decoration={brace,amplitude=2mm,raise=-1mm}]
\tikzstyle{small braceedge}=[decorate,decoration={brace,amplitude=1mm,raise=-1mm}]

\tikzstyle{doubled}=[line width=1.6pt] 
\tikzstyle{boldedge}=[doubled,shorten <=-0.17mm,shorten >=-0.17mm]
\tikzstyle{boldedgegray}=[doubled,gray,shorten <=-0.17mm,shorten >=-0.17mm]
\tikzstyle{singleedgegray}=[gray]

\tikzstyle{semidoubled}=[line width=1.4pt] 
\tikzstyle{semiboldedgegray}=[semidoubled,gray,shorten <=-0.17mm,shorten >=-0.17mm]

\tikzstyle{boxedge}=[semiboldedgegray]

\tikzstyle{boldedgedashed}=[very thick,dashed,shorten <=-0.17mm,shorten >=-0.17mm]
\tikzstyle{vboldedgedashed}=[doubled,dashed,shorten <=-0.17mm,shorten >=-0.17mm]
\tikzstyle{left hook arrow}=[left hook-latex]
\tikzstyle{right hook arrow}=[right hook-latex]
\tikzstyle{sembracket}=[line width=0.5pt,shorten <=-0.07mm,shorten >=-0.07mm]

\tikzstyle{causal edge}=[->,thick,gray]
\tikzstyle{causal nondir}=[thick,gray]
\tikzstyle{timeline}=[thick,gray, dashed]

\tikzstyle{cedge}=[<->,thick,gray!70!white]

\tikzstyle{empty diagram}=[draw=gray!40!white,dashed,shape=rectangle,minimum width=1cm,minimum height=1cm]
\tikzstyle{empty diagram small}=[draw=gray!50!white,dashed,shape=rectangle,minimum width=0.6cm,minimum height=0.5cm]

\tikzstyle{dot}=[inner sep=0mm,minimum width=2mm,minimum height=2mm,draw,shape=circle]
\tikzstyle{bigdot}=[inner sep=0mm,minimum width=5mm,minimum height=5mm,draw,shape=circle]
\tikzstyle{leak}=[white dot, shape=regular polygon, minimum size=3.3 mm, regular polygon sides=3, outer sep=-0.2mm, regular polygon rotate=270]
\tikzstyle{proj}=[regular polygon,regular polygon sides=4,draw,scale=0.75,inner sep=-0.5pt,minimum width=6mm,fill=white]
\tikzstyle{projOut}=[regular polygon,regular polygon sides=3,draw,scale=0.75,inner sep=-0.5pt,minimum width=7.5mm,fill=white,regular polygon rotate=180]
\tikzstyle{projIn}=[regular polygon,regular polygon sides=3,draw,scale=0.75,inner sep=-0.5pt,minimum width=7.5mm,fill=white]
\tikzstyle{Vleak}=[white dot, shape=regular polygon, minimum size=3.3 mm, regular polygon sides=3, outer sep=-0.2mm, regular polygon rotate=90]
\tikzstyle{dleak}=[white dot, line width=1.6pt, shape=regular polygon, minimum size=3.3 mm, regular polygon sides=3, outer sep=-0.2mm, regular polygon rotate=270]

\tikzstyle{Wsquare}=[white dot, shape=regular polygon, rounded corners=0.8 mm, minimum size=3.3 mm, regular polygon sides=3, outer sep=-0.2mm]
\tikzstyle{Wsquareadj}=[white dot, shape=regular polygon, rounded corners=0.8 mm, minimum size=3.3 mm, regular polygon sides=3, outer sep=-0.2mm, regular polygon rotate=180]
\tikzstyle{ddot}=[inner sep=0mm, doubled, minimum width=2.5mm,minimum height=2.5mm,draw,shape=circle]

\tikzstyle{clear dot}=[dot,fill=none,text depth=-0.2mm,draw=gray, line width = .75pt]
\tikzstyle{tall clear dot}=[dot,fill=none,text depth=-0.2mm,draw=gray, line width = .75pt,shape=ellipse, minimum height=5mm]
\tikzstyle{wide clear dot}=[dot,fill=none,text depth=-0.2mm,draw=gray, line width = .75pt, shape=ellipse, minimum width = 5mm]
\tikzstyle{very wide clear dot}=[dot,fill=none,text depth=-0.2mm,draw=gray, line width = .75pt, shape=ellipse, minimum width = 7mm ]

\tikzstyle{black dot}=[dot,fill=black]
\tikzstyle{white dot}=[dot,fill=white,,text depth=-0.2mm]
\tikzstyle{white Wsquare}=[Wsquare,fill=gray,,text depth=-0.2mm]
\tikzstyle{white Wsquareadj}=[Wsquareadj,fill=white,,text depth=-0.2mm]
\tikzstyle{green dot}=[white dot] 
\tikzstyle{gray dot}=[dot,fill=gray!40!white,,text depth=-0.2mm]
\tikzstyle{red dot}=[gray dot] 

\tikzstyle{black ddot}=[ddot,fill=black]
\tikzstyle{white ddot}=[ddot,fill=white]
\tikzstyle{gray ddot}=[ddot,fill=gray!40!white]

\tikzstyle{gray edge}=[gray!60!white]

\tikzstyle{small dot}=[inner sep=0.2mm,minimum width=0pt,minimum height=0pt,draw,shape=circle]

\tikzstyle{small black dot}=[small dot,fill=black]
\tikzstyle{small white dot}=[small dot,fill=white]
\tikzstyle{small gray dot}=[small dot,fill=gray,draw=gray]

\tikzstyle{causal dot}=[inner sep=0.4mm,minimum width=0pt,minimum height=0pt,draw=white,shape=circle,fill=gray!40!white]

\tikzstyle{phase dimensions}=[minimum size=5mm,font=\footnotesize,rectangle,rounded corners=2.5mm,inner sep=0.2mm,outer sep=-2mm]
\tikzstyle{dphase dimensions}=[minimum size=5mm,font=\footnotesize,rectangle,rounded corners=2.5mm,inner sep=0.2mm,outer sep=-2mm]

\tikzstyle{white phase dot}=[dot,fill=white,phase dimensions]
\tikzstyle{white phase ddot}=[ddot,fill=white,dphase dimensions]

\tikzstyle{white rect ddot}=[draw=black,fill=white,doubled,minimum size=5mm,font=\footnotesize,rectangle,rounded corners=2.5mm,inner sep=0.2mm]
\tikzstyle{gray rect ddot}=[draw=black,fill=gray!40!white,doubled,minimum size=6mm,font=\footnotesize,rectangle,rounded corners=3mm]

\tikzstyle{gray phase dot}=[dot,fill=gray!40!white,phase dimensions]
\tikzstyle{gray phase ddot}=[ddot,fill=gray!40!white,dphase dimensions]
\tikzstyle{grey phase dot}=[gray phase dot]
\tikzstyle{grey phase ddot}=[gray phase ddot]

\tikzstyle{small phase dimensions}=[minimum size=4mm,font=\tiny,rectangle,rounded corners=2mm,inner sep=0.2mm,outer sep=-2mm]
\tikzstyle{small dphase dimensions}=[minimum size=4mm,font=\tiny,rectangle,rounded corners=2mm,inner sep=0.2mm,outer sep=-2mm]

\tikzstyle{small gray phase dot}=[dot,fill=gray!40!white,small phase dimensions]
\tikzstyle{small gray phase ddot}=[ddot,fill=gray!40!white,small dphase dimensions]

\tikzstyle{small map}=[draw,shape=rectangle,minimum height=4mm,minimum width=4mm,fill=white]

\tikzstyle{cnot}=[fill=white,shape=circle,inner sep=-1.4pt]

\tikzstyle{asym hadamard}=[fill=white,draw,shape=NEbox,inner sep=0.6mm,font=\footnotesize,minimum height=4mm]
\tikzstyle{asym hadamard conj}=[fill=white,draw,shape=NWbox,inner sep=0.6mm,font=\footnotesize,minimum height=4mm]
\tikzstyle{asym hadamard dag}=[fill=white,draw,shape=SEbox,inner sep=0.6mm,font=\footnotesize,minimum height=4mm]

\tikzstyle{hadamard}=[fill=white,draw,inner sep=0.6mm,font=\footnotesize,minimum height=4mm,minimum width=4mm]
\tikzstyle{small hadamard}=[fill=white,draw,inner sep=0.6mm,minimum height=1.5mm,minimum width=1.5mm]
\tikzstyle{small hadamard rotate}=[small hadamard,rotate=45]
\tikzstyle{dhadamard}=[hadamard,doubled]
\tikzstyle{small dhadamard}=[small hadamard,doubled]
\tikzstyle{small dhadamard rotate}=[small hadamard rotate,doubled]
\tikzstyle{antipode}=[white dot,inner sep=0.3mm,font=\footnotesize]

\tikzstyle{scalar}=[diamond,draw,inner sep=0.5pt,font=\small]
\tikzstyle{dscalar}=[diamond,doubled, draw,inner sep=0.5pt,font=\small]

\tikzstyle{small box}=[rectangle,inline text,fill=white,draw,minimum height=5mm,yshift=-0.5mm,minimum width=5mm,font=\small]
\tikzstyle{small gray box}=[small box,fill=gray!30]
\tikzstyle{medium box}=[rectangle,inline text,fill=white,draw,minimum height=5mm,yshift=-0.5mm,minimum width=10mm,font=\small]
\tikzstyle{square box}=[small box] 
\tikzstyle{medium gray box}=[small box,fill=gray!30]
\tikzstyle{semilarge box}=[rectangle,inline text,fill=white,draw,minimum height=5mm,yshift=-0.5mm,minimum width=12.5mm,font=\small]
\tikzstyle{large box}=[rectangle,inline text,fill=white,draw,minimum height=5mm,yshift=-0.5mm,minimum width=15mm,font=\small]
\tikzstyle{large gray box}=[small box,fill=gray!30]

\tikzstyle{Bayes box}=[rectangle,fill=black,draw, minimum height=3mm, minimum width=3mm]

\tikzstyle{gray square point}=[small box,fill=gray!50]

\tikzstyle{dphase box white}=[dhadamard]
\tikzstyle{dphase box gray}=[dhadamard,fill=gray!50!white]
\tikzstyle{phase box white}=[hadamard]
\tikzstyle{phase box gray}=[hadamard,fill=gray!50!white]

\tikzstyle{point}=[regular polygon,regular polygon sides=3,draw,scale=0.75,inner sep=-0.5pt,minimum width=9mm,fill=white,regular polygon rotate=180]
\tikzstyle{infpoint}=[regular polygon,regular polygon sides=3,draw,scale=0.75,inner sep=-0.5pt,minimum width=9mm,fill=white,regular polygon rotate=90]
\tikzstyle{point nosep}=[regular polygon,regular polygon sides=3,draw,scale=0.75,inner sep=-2pt,minimum width=9mm,fill=white,regular polygon rotate=180]
\tikzstyle{infcopoint}=[regular polygon,regular polygon sides=3,draw,scale=0.75,inner sep=-0.5pt,minimum width=9mm,fill=white,regular polygon rotate=270]
\tikzstyle{copoint}=[regular polygon,regular polygon sides=3,draw,scale=0.75,inner sep=-0.5pt,minimum width=9mm,fill=white]
\tikzstyle{dpoint}=[point,doubled]
\tikzstyle{dcopoint}=[copoint,doubled]

\tikzstyle{pointgrow}=[shape=cornerpoint,kpoint common,scale=0.75,inner sep=3pt]
\tikzstyle{pointgrow dag}=[shape=cornercopoint,kpoint common,scale=0.75,inner sep=3pt]

\tikzstyle{wide copoint}=[fill=white,draw,shape=isosceles triangle,shape border rotate=90,isosceles triangle stretches=true,inner sep=0pt,minimum width=1.5cm,minimum height=6.12mm]
\tikzstyle{wide point}=[fill=white,draw,shape=isosceles triangle,shape border rotate=-90,isosceles triangle stretches=true,inner sep=0pt,minimum width=1.5cm,minimum height=6.12mm,yshift=-0.0mm]
\tikzstyle{wide point plus}=[fill=white,draw,shape=isosceles triangle,shape border rotate=-90,isosceles triangle stretches=true,inner sep=0pt,minimum width=1.74cm,minimum height=7mm,yshift=-0.0mm]

\tikzstyle{wide dpoint}=[fill=white,doubled,draw,shape=isosceles triangle,shape border rotate=-90,isosceles triangle stretches=true,inner sep=0pt,minimum width=1.5cm,minimum height=6.12mm,yshift=-0.0mm]

\tikzstyle{tinypoint}=[regular polygon,regular polygon sides=3,draw,scale=0.55,inner sep=-0.15pt,minimum width=6mm,fill=white,regular polygon rotate=180]

\tikzstyle{white point}=[point]
\tikzstyle{white dpoint}=[dpoint]
\tikzstyle{green point}=[white point] 
\tikzstyle{white copoint}=[copoint]
\tikzstyle{gray point}=[point,fill=gray!40!white]
\tikzstyle{gray dpoint}=[gray point,doubled]
\tikzstyle{red point}=[gray point] 
\tikzstyle{gray copoint}=[copoint,fill=gray!40!white]
\tikzstyle{gray dcopoint}=[gray copoint,doubled]

\tikzstyle{white point guide}=[regular polygon,regular polygon sides=3,font=\scriptsize,draw,scale=0.65,inner sep=-0.5pt,minimum width=9mm,fill=white,regular polygon rotate=180]

\tikzstyle{black point}=[point,fill=black,font=\color{white}]
\tikzstyle{black copoint}=[copoint,fill=black,font=\color{white}]

\tikzstyle{tiny gray point}=[tinypoint,fill=gray!40!white]

\tikzstyle{diredge}=[->]
\tikzstyle{ddiredge}=[<->]
\tikzstyle{rdiredge}=[<-]
\tikzstyle{thickdiredge}=[->, very thick]
\tikzstyle{pointer edge}=[->,very thick,gray]
\tikzstyle{pointer edge part}=[very thick,gray]
\tikzstyle{dashed edge}=[dashed]
\tikzstyle{thick dashed edge}=[very thick,dashed]
\tikzstyle{thick gray dashed edge}=[thick dashed edge,gray!40]
\tikzstyle{thick map edge}=[very thick,|->]

\makeatletter
\newcommand{\boxshape}[3]{%
\pgfdeclareshape{#1}{
\inheritsavedanchors[from=rectangle] 
\inheritanchorborder[from=rectangle]
\inheritanchor[from=rectangle]{center}
\inheritanchor[from=rectangle]{north}
\inheritanchor[from=rectangle]{south}
\inheritanchor[from=rectangle]{west}
\inheritanchor[from=rectangle]{east}
\backgroundpath{
\southwest \pgf@xa=\pgf@x \pgf@ya=\pgf@y
\northeast \pgf@xb=\pgf@x \pgf@yb=\pgf@y

\@tempdima=#2
\@tempdimb=#3

\pgfpathmoveto{\pgfpoint{\pgf@xa - 5pt + \@tempdima}{\pgf@ya}}
\pgfpathlineto{\pgfpoint{\pgf@xa - 5pt - \@tempdima}{\pgf@yb}}
\pgfpathlineto{\pgfpoint{\pgf@xb + 5pt + \@tempdimb}{\pgf@yb}}
\pgfpathlineto{\pgfpoint{\pgf@xb + 5pt - \@tempdimb}{\pgf@ya}}
\pgfpathlineto{\pgfpoint{\pgf@xa - 5pt + \@tempdima}{\pgf@ya}}
\pgfpathclose
}
}}

\boxshape{NEbox}{0pt}{3pt}
\boxshape{SEbox}{0pt}{-3pt}
\boxshape{NWbox}{5pt}{0pt}
\boxshape{SWbox}{-5pt}{0pt}
\boxshape{EBox}{-3pt}{3pt}
\boxshape{WBox}{3pt}{-3pt}
\makeatother

\tikzstyle{cloud}=[shape=cloud,draw,minimum width=1.5cm,minimum height=1.5cm]

\tikzstyle{map}=[draw,shape=NEbox,inner sep=1pt,minimum height=4mm,fill=white]
\tikzstyle{dashedmap}=[draw,dashed,shape=NEbox,inner sep=2pt,minimum height=6mm,fill=white]
\tikzstyle{mapdag}=[draw,shape=SEbox,inner sep=1pt,minimum height=4mm,fill=white]
\tikzstyle{mapadj}=[draw,shape=SEbox,inner sep=2pt,minimum height=6mm,fill=white]
\tikzstyle{maptrans}=[draw,shape=SWbox,inner sep=2pt,minimum height=6mm,fill=white]
\tikzstyle{mapconj}=[draw,shape=NWbox,inner sep=2pt,minimum height=6mm,fill=white]

\tikzstyle{medium map}=[draw,shape=NEbox,inner sep=2pt,minimum height=6mm,fill=white,minimum width=7mm]
\tikzstyle{medium map dag}=[draw,shape=SEbox,inner sep=2pt,minimum height=6mm,fill=white,minimum width=7mm]
\tikzstyle{medium map adj}=[draw,shape=SEbox,inner sep=2pt,minimum height=6mm,fill=white,minimum width=7mm]
\tikzstyle{medium map trans}=[draw,shape=SWbox,inner sep=2pt,minimum height=6mm,fill=white,minimum width=7mm]
\tikzstyle{medium map conj}=[draw,shape=NWbox,inner sep=2pt,minimum height=6mm,fill=white,minimum width=7mm]
\tikzstyle{semilarge map}=[draw,shape=NEbox,inner sep=2pt,minimum height=6mm,fill=white,minimum width=9.5mm]
\tikzstyle{semilarge map trans}=[draw,shape=SWbox,inner sep=2pt,minimum height=6mm,fill=white,minimum width=9.5mm]
\tikzstyle{semilarge map adj}=[draw,shape=SEbox,inner sep=2pt,minimum height=6mm,fill=white,minimum width=9.5mm]
\tikzstyle{semilarge map dag}=[draw,shape=SEbox,inner sep=2pt,minimum height=6mm,fill=white,minimum width=9.5mm]
\tikzstyle{semilarge map conj}=[draw,shape=NWbox,inner sep=2pt,minimum height=6mm,fill=white,minimum width=9.5mm]
\tikzstyle{large map}=[draw,shape=NEbox,inner sep=2pt,minimum height=6mm,fill=white,minimum width=12mm]
\tikzstyle{large map conj}=[draw,shape=NWbox,inner sep=2pt,minimum height=6mm,fill=white,minimum width=12mm]
\tikzstyle{very large map}=[draw,shape=NEbox,inner sep=2pt,minimum height=6mm,fill=white,minimum width=17mm]

\tikzstyle{medium dmap}=[draw,doubled,shape=NEbox,inner sep=2pt,minimum height=6mm,fill=white,minimum width=7mm]
\tikzstyle{medium dmap dag}=[draw,doubled,shape=SEbox,inner sep=2pt,minimum height=6mm,fill=white,minimum width=7mm]
\tikzstyle{medium dmap adj}=[draw,doubled,shape=SEbox,inner sep=2pt,minimum height=6mm,fill=white,minimum width=7mm]
\tikzstyle{medium dmap trans}=[draw,doubled,shape=SWbox,inner sep=2pt,minimum height=6mm,fill=white,minimum width=7mm]
\tikzstyle{medium dmap conj}=[draw,doubled,shape=NWbox,inner sep=2pt,minimum height=6mm,fill=white,minimum width=7mm]
\tikzstyle{semilarge dmap}=[draw,doubled,shape=NEbox,inner sep=2pt,minimum height=6mm,fill=white,minimum width=9.5mm]
\tikzstyle{semilarge dmap trans}=[draw,doubled,shape=SWbox,inner sep=2pt,minimum height=6mm,fill=white,minimum width=9.5mm]
\tikzstyle{semilarge dmap adj}=[draw,doubled,shape=SEbox,inner sep=2pt,minimum height=6mm,fill=white,minimum width=9.5mm]
\tikzstyle{semilarge dmap dag}=[draw,doubled,shape=SEbox,inner sep=2pt,minimum height=6mm,fill=white,minimum width=9.5mm]
\tikzstyle{semilarge dmap conj}=[draw,doubled,shape=NWbox,inner sep=2pt,minimum height=6mm,fill=white,minimum width=9.5mm]
\tikzstyle{large dmap}=[draw,doubled,shape=NEbox,inner sep=2pt,minimum height=6mm,fill=white,minimum width=12mm]
\tikzstyle{large dmap conj}=[draw,doubled,shape=NWbox,inner sep=2pt,minimum height=6mm,fill=white,minimum width=12mm]
\tikzstyle{large dmap trans}=[draw,doubled,shape=SWbox,inner sep=2pt,minimum height=6mm,fill=white,minimum width=12mm]
\tikzstyle{large dmap adj}=[draw,doubled,shape=SEbox,inner sep=2pt,minimum height=6mm,fill=white,minimum width=12mm]
\tikzstyle{large dmap dag}=[draw,doubled,shape=SEbox,inner sep=2pt,minimum height=6mm,fill=white,minimum width=12mm]
\tikzstyle{very large dmap}=[draw,doubled,shape=NEbox,inner sep=2pt,minimum height=6mm,fill=white,minimum width=19.5mm]

\tikzstyle{muxbox}=[draw,shape=rectangle,minimum height=3mm,minimum width=3mm,fill=white]
\tikzstyle{dmuxbox}=[muxbox,doubled]

\tikzstyle{box}=[draw,shape=rectangle,inner sep=2pt,minimum height=6mm,minimum width=6mm,fill=white]
\tikzstyle{dbox}=[draw,doubled,shape=rectangle,inner sep=2pt,minimum height=6mm,minimum width=6mm,fill=white]
\tikzstyle{dmap}=[draw,doubled,shape=NEbox,inner sep=2pt,minimum height=6mm,fill=white]
\tikzstyle{dmapdag}=[draw,doubled,shape=SEbox,inner sep=2pt,minimum height=6mm,fill=white]
\tikzstyle{dmapadj}=[draw,doubled,shape=SEbox,inner sep=2pt,minimum height=6mm,fill=white]
\tikzstyle{dmaptrans}=[draw,doubled,shape=SWbox,inner sep=2pt,minimum height=6mm,fill=white]
\tikzstyle{dmapconj}=[draw,doubled,shape=NWbox,inner sep=2pt,minimum height=6mm,fill=white]

\tikzstyle{ddmap}=[draw,doubled,dashed,shape=NEbox,inner sep=2pt,minimum height=6mm,fill=white]
\tikzstyle{ddmapdag}=[draw,doubled,dashed,shape=SEbox,inner sep=2pt,minimum height=6mm,fill=white]
\tikzstyle{ddmapadj}=[draw,doubled,dashed,shape=SEbox,inner sep=2pt,minimum height=6mm,fill=white]
\tikzstyle{ddmaptrans}=[draw,doubled,dashed,shape=SWbox,inner sep=2pt,minimum height=6mm,fill=white]
\tikzstyle{ddmapconj}=[draw,doubled,dashed,shape=NWbox,inner sep=2pt,minimum height=6mm,fill=white]

\boxshape{sNEbox}{0pt}{3pt}
\boxshape{sSEbox}{0pt}{-3pt}
\boxshape{sNWbox}{3pt}{0pt}
\boxshape{sSWbox}{-3pt}{0pt}
\tikzstyle{smap}=[draw,shape=sNEbox,fill=white]
\tikzstyle{smapdag}=[draw,shape=sSEbox,fill=white]
\tikzstyle{smapadj}=[draw,shape=sSEbox,fill=white]
\tikzstyle{smaptrans}=[draw,shape=sSWbox,fill=white]
\tikzstyle{smapconj}=[draw,shape=sNWbox,fill=white]

\tikzstyle{dsmap}=[draw,dashed,shape=sNEbox,fill=white]
\tikzstyle{dsmapdag}=[draw,dashed,shape=sSEbox,fill=white]
\tikzstyle{dsmaptrans}=[draw,dashed,shape=sSWbox,fill=white]
\tikzstyle{dsmapconj}=[draw,dashed,shape=sNWbox,fill=white]

\boxshape{mNEbox}{0pt}{10pt}
\boxshape{mSEbox}{0pt}{-10pt}
\boxshape{mNWbox}{10pt}{0pt}
\boxshape{mSWbox}{-10pt}{0pt}
\tikzstyle{mmap}=[draw,shape=mNEbox]
\tikzstyle{mmapdag}=[draw,shape=mSEbox]
\tikzstyle{mmaptrans}=[draw,shape=mSWbox]
\tikzstyle{mmapconj}=[draw,shape=mNWbox]

\tikzstyle{mmapgray}=[draw,fill=gray!40!white,shape=mNEbox]
\tikzstyle{smapgray}=[draw,fill=gray!40!white,shape=sNEbox]

\makeatletter

\pgfdeclareshape{cornerpoint}{
\inheritsavedanchors[from=rectangle] 
\inheritanchorborder[from=rectangle]
\inheritanchor[from=rectangle]{center}
\inheritanchor[from=rectangle]{north}
\inheritanchor[from=rectangle]{south}
\inheritanchor[from=rectangle]{west}
\inheritanchor[from=rectangle]{east}
\backgroundpath{
\southwest \pgf@xa=\pgf@x \pgf@ya=\pgf@y
\northeast \pgf@xb=\pgf@x \pgf@yb=\pgf@y

\pgfmathsetmacro{\pgf@shorten@left}{\pgfkeysvalueof{/tikz/shorten left}}
\pgfmathsetmacro{\pgf@shorten@right}{\pgfkeysvalueof{/tikz/shorten right}}

\pgfpathmoveto{\pgfpoint{0.5 * (\pgf@xa + \pgf@xb)}{\pgf@ya - 5pt}}
\pgfpathlineto{\pgfpoint{\pgf@xa - 8pt + \pgf@shorten@left}{\pgf@yb - 1.5 * \pgf@shorten@left}}
\pgfpathlineto{\pgfpoint{\pgf@xa - 8pt + \pgf@shorten@left}{\pgf@yb}}
\pgfpathlineto{\pgfpoint{\pgf@xb + 8pt - \pgf@shorten@right}{\pgf@yb}}
\pgfpathlineto{\pgfpoint{\pgf@xb + 8pt - \pgf@shorten@right}{\pgf@yb - 1.5 * \pgf@shorten@right}}
\pgfpathclose
}
}

\pgfdeclareshape{cornercopoint}{
\inheritsavedanchors[from=rectangle] 
\inheritanchorborder[from=rectangle]
\inheritanchor[from=rectangle]{center}
\inheritanchor[from=rectangle]{north}
\inheritanchor[from=rectangle]{south}
\inheritanchor[from=rectangle]{west}
\inheritanchor[from=rectangle]{east}
\backgroundpath{
\southwest \pgf@xa=\pgf@x \pgf@ya=\pgf@y
\northeast \pgf@xb=\pgf@x \pgf@yb=\pgf@y

\pgfmathsetmacro{\pgf@shorten@left}{\pgfkeysvalueof{/tikz/shorten left}}
\pgfmathsetmacro{\pgf@shorten@right}{\pgfkeysvalueof{/tikz/shorten right}}

\pgfpathmoveto{\pgfpoint{0.5 * (\pgf@xa + \pgf@xb)}{\pgf@yb + 5pt}}
\pgfpathlineto{\pgfpoint{\pgf@xa - 8pt + \pgf@shorten@left}{\pgf@ya + 1.5 * \pgf@shorten@left}}
\pgfpathlineto{\pgfpoint{\pgf@xa - 8pt + \pgf@shorten@left}{\pgf@ya}}
\pgfpathlineto{\pgfpoint{\pgf@xb + 8pt - \pgf@shorten@right}{\pgf@ya}}
\pgfpathlineto{\pgfpoint{\pgf@xb + 8pt - \pgf@shorten@right}{\pgf@ya + 1.5 * \pgf@shorten@right}}
\pgfpathclose
}
}

\makeatother

\pgfkeyssetvalue{/tikz/shorten left}{0pt}
\pgfkeyssetvalue{/tikz/shorten right}{0pt}

\tikzstyle{kpoint common}=[draw,fill=white,inner sep=1pt,minimum height=4mm]
\tikzstyle{kpoint sc}=[shape=cornerpoint,kpoint common]
\tikzstyle{kpoint adjoint sc}=[shape=cornercopoint,kpoint common]
\tikzstyle{kpoint}=[shape=cornerpoint,shorten left=5pt,kpoint common]
\tikzstyle{kpoint adjoint}=[shape=cornercopoint,shorten left=5pt,kpoint common]
\tikzstyle{kpoint conjugate}=[shape=cornerpoint,shorten right=5pt,kpoint common]
\tikzstyle{kpoint transpose}=[shape=cornercopoint,shorten right=5pt,kpoint common]
\tikzstyle{kpoint symm}=[shape=cornerpoint,shorten left=5pt,shorten right=5pt,kpoint common]

\tikzstyle{wide kpoint sc}=[shape=cornerpoint,kpoint common, minimum width=1 cm]
\tikzstyle{wide kpointdag sc}=[shape=cornercopoint,kpoint common, minimum width=1 cm]

\tikzstyle{black kpoint}=[shape=cornerpoint,shorten left=5pt,kpoint common,fill=black,font=\color{white}]

\tikzstyle{black kpoint sm}=[shape=cornerpoint,shorten left=5pt,kpoint common,fill=black,font=\color{white},scale=0.75]

\tikzstyle{black kpoint adjoint}=[shape=cornercopoint,shorten left=5pt,kpoint common,fill=black,font=\color{white}]
\tikzstyle{black kpointadj}=[shape=cornercopoint,shorten left=5pt,kpoint common,fill=black,font=\color{white}]

\tikzstyle{black kpointadj sm}=[shape=cornercopoint,shorten left=5pt,kpoint common,fill=black,font=\color{white},scale=0.75]

\tikzstyle{black dkpoint}=[shape=cornerpoint,shorten left=5pt,kpoint common,fill=black, doubled,font=\color{white}]
\tikzstyle{black dkpoint adjoint}=[shape=cornercopoint,shorten left=5pt,kpoint common,fill=black, doubled,font=\color{white}]
\tikzstyle{black dkpointadj}=[shape=cornercopoint,shorten left=5pt,kpoint common,fill=black, doubled,font=\color{white}]

\tikzstyle{black dkpoint sm}=[shape=cornerpoint,shorten left=5pt,kpoint common,fill=black, doubled,font=\color{white},scale=0.75]
\tikzstyle{black dkpointadj sm}=[shape=cornercopoint,shorten left=5pt,kpoint common,fill=black, doubled,font=\color{white},scale=0.75]

\tikzstyle{kpointdag}=[kpoint adjoint]
\tikzstyle{kpointadj}=[kpoint adjoint]
\tikzstyle{kpointconj}=[kpoint conjugate]
\tikzstyle{kpointtrans}=[kpoint transpose]

\tikzstyle{big kpoint}=[kpoint, minimum width=1.2 cm, minimum height=8mm, inner sep=4pt, text depth=3mm]

\tikzstyle{wide kpoint}=[kpoint, minimum width=1 cm, inner sep=2pt]
\tikzstyle{wide kpointdag}=[kpointdag, minimum width=1 cm, inner sep=2pt]
\tikzstyle{wide kpointconj}=[kpointconj, minimum width=1 cm, inner sep=2pt]
\tikzstyle{wide kpointtrans}=[kpointtrans, minimum width=1 cm, inner sep=2pt]

\tikzstyle{wider kpoint}=[kpoint, minimum width=1.25 cm, inner sep=2pt]
\tikzstyle{wider kpointdag}=[kpointdag, minimum width=1.25 cm, inner sep=2pt]
\tikzstyle{wider kpointconj}=[kpointconj, minimum width=1.25 cm, inner sep=2pt]
\tikzstyle{wider kpointtrans}=[kpointtrans, minimum width=1.25 cm, inner sep=2pt]

\tikzstyle{gray kpoint}=[kpoint,fill=gray!50!white]
\tikzstyle{gray kpointdag}=[kpointdag,fill=gray!50!white]
\tikzstyle{gray kpointadj}=[kpointadj,fill=gray!50!white]
\tikzstyle{gray kpointconj}=[kpointconj,fill=gray!50!white]
\tikzstyle{gray kpointtrans}=[kpointtrans,fill=gray!50!white]

\tikzstyle{gray dkpoint}=[kpoint,fill=gray!50!white,doubled]
\tikzstyle{gray dkpointdag}=[kpointdag,fill=gray!50!white,doubled]
\tikzstyle{gray dkpointadj}=[kpointadj,fill=gray!50!white,doubled]
\tikzstyle{gray dkpointconj}=[kpointconj,fill=gray!50!white,doubled]
\tikzstyle{gray dkpointtrans}=[kpointtrans,fill=gray!50!white,doubled]

\tikzstyle{white label}=[draw,fill=white,rectangle,inner sep=0.7 mm]
\tikzstyle{gray label}=[draw,fill=gray!50!white,rectangle,inner sep=0.7 mm]
\tikzstyle{black label}=[draw,fill=black,rectangle,inner sep=0.7 mm]

\tikzstyle{dkpoint}=[kpoint,doubled]
\tikzstyle{wide dkpoint}=[wide kpoint,doubled]
\tikzstyle{dkpointdag}=[kpoint adjoint,doubled]
\tikzstyle{wide dkpointdag}=[wide kpointdag,doubled]
\tikzstyle{dkcopoint}=[kpoint adjoint,doubled]
\tikzstyle{dkpointadj}=[kpoint adjoint,doubled]
\tikzstyle{dkpointconj}=[kpoint conjugate,doubled]
\tikzstyle{dkpointtrans}=[kpoint transpose,doubled]

\tikzstyle{kscalar}=[kpoint common, shape=EBox, inner xsep=-1pt, inner ysep=3pt,font=\small]
\tikzstyle{kscalarconj}=[kpoint common, shape=WBox, inner xsep=-1pt, inner ysep=3pt,font=\small]

\tikzstyle{spekpoint}=[kpoint sc,minimum height=5mm,inner sep=3pt]
\tikzstyle{spekcopoint}=[kpoint adjoint sc,minimum height=5mm,inner sep=3pt]

\tikzstyle{dspekpoint}=[spekpoint,doubled]
\tikzstyle{dspekcopoint}=[spekcopoint,doubled]

 \tikzstyle{upground}=[circuit ee IEC,thick,ground,rotate=90,scale=2.5]
 \tikzstyle{downground}=[circuit ee IEC,thick,ground,rotate=-90,scale=2.5]
 \tikzstyle{infupground}=[circuit ee IEC,thick,ground,rotate=0,scale=2.5]
 \tikzstyle{infdownground}=[circuit ee IEC,thick,ground,rotate=180,scale=2.5]
 \tikzstyle{bigground}=[regular polygon,regular polygon sides=3,draw=gray,scale=0.50,inner sep=-0.5pt,minimum width=10mm,fill=gray]

\tikzstyle{arrs}=[-latex,font=\small,auto]
\tikzstyle{arrow plain}=[arrs]
\tikzstyle{arrow dashed}=[dashed,arrs]
\tikzstyle{arrow bold}=[very thick,arrs]
\tikzstyle{arrow hide}=[draw=white!0,-]
\tikzstyle{arrow reverse}=[latex-]
\tikzstyle{cdnode}=[]

\tikzstyle{tilde}=[draw=blue]
\tikzstyle{tildelabel}=[text=blue]

\let\olddagger\dagger
\renewcommand{\dagger}{\ensuremath{\olddagger}\xspace}

\usepackage{makeidx}
\makeindex

\theoremstyle{plain}
\newtheorem*{main theorem}{Main Theorem}
\newtheorem{theorem}{Theorem}[section]

\newtheorem{lemma}[theorem]{Lemma}

\newtheorem{definition}[theorem]{Definition}

\newtheorem{example*}[theorem]{Example*}
\newtheorem{examples*}[theorem]{Examples*}

\newtheorem{remark*}[theorem]{Remark*}

\newtheorem*{search problem}{Search Problem}

\usepackage{color}
\def\bR{\begin{color}{red}}
\def\bB{\begin{color}{blue}}
\def\bM{\begin{color}{magenta}}
\def\bC{\begin{color}{cyan}}
\def\bW{\begin{color}{white}}
\def\bBl{\begin{color}{black}}
\def\bG{\begin{color}{green}}
\def\bY{\begin{color}{yellow}}
\def\e{\end{color}\xspace}
\newcommand{\bit}{\begin{itemize}}
\newcommand{\eit}{\end{itemize}\par\noindent}
\newcommand{\ben}{\begin{enumerate}}
\newcommand{\een}{\end{enumerate}\par\noindent}
\newcommand{\beq}{\begin{equation}}
\newcommand{\eeq}{\end{equation}\par\noindent}
\newcommand{\beqa}{\begin{eqnarray*}}
\newcommand{\eeqa}{\end{eqnarray*}\par\noindent}
\newcommand{\beqn}{\begin{eqnarray}}
\newcommand{\eeqn}{\end{eqnarray}\par\noindent}

\def\jR{\begin{color}{black}}
\def\jB{\begin{color}{black}}
\def\jM{\begin{color}{magenta}}
\def\jC{\begin{color}{cyan}}
\def\jW{\begin{color}{white}}
\def\jBl{\begin{color}{black}}
\def\jG{\begin{color}{green}}
\def\jY{\begin{color}{yellow}}

\newcommand{\xiNC}{\xi_{\rm{\kern -0.8pt n \kern -0.7pt c}}}
\usepackage{soul}
\begin{document}
\title{Noncontextual ontological models of operational probabilistic theories}
\author{Sina Soltani}
\email{sina.soltani@phdstud.ug.edu.pl}
\affiliation{International Centre for Theory of Quantum Technologies, Uniwersytet Gdański, ul.~Jana Bażyńskiego 1A, 80-309 Gdańsk, Polska}
\author{Marco Erba
}
\affiliation{International Centre for Theory of Quantum Technologies, Uniwersytet Gdański, ul.~Jana Bażyńskiego 1A, 80-309 Gdańsk, Polska}
\author{David Schmid}
\affiliation{Perimeter Institute for Theoretical Physics, 31 Caroline Street North, Waterloo, Ontario Canada N2L 2Y5}
\author{John H. Selby}
\affiliation{International Centre for Theory of Quantum Technologies, Uniwersytet Gdański, ul.~Jana Bażyńskiego 1A, 80-309 Gdańsk, Polska}
\affiliation{Theoretical Sciences Visiting Program, Okinawa Institute of Science
and Technology Graduate University, Onna, 904-0495, Japan}
\begin{abstract}
An experiment or theory is classically explainable if it can be reproduced by some noncontextual ontological model. In this work, we adapt the notion of ontological models and generalized noncontextuality so it applies to the framework of operational probabilistic theories (OPTs).  A defining feature of quotiented OPTs, which sets them apart from the closely related framework of generalized probabilistic theories (GPTs), is their explicit specification of the structure of instruments, these being generalizations of \emph{quantum instruments} (including nondestructive measurements); in particular, one needs to explicitly declare which collections of transformations constitute a valid instrument. We are particularly interested in strongly causal OPTs, in which the choice of a future instrument can be conditioned on a past measurement outcome. This instrument structure might seem to permit the possibility of a contextual kind of ontological representation, where the representation of a given transformation depends on which instrument it is considered a part of. However, we prove that this is not possible by showing that for strongly causal quotiented OPTs the structure of instruments does {\em not} allow for such a contextual ontological representation.  It follows that ontological representations of strongly causal quotiented OPTs are entirely determined by their action on individual transformations, with no dependence on the structure of instruments. 

\end{abstract}
\maketitle

{
  \hypersetup{linkcolor=purple}
  \tableofcontents
}

\section{Introduction}

If an experiment or theory can be reproduced within a generalized-noncontextual ontological model, then it is said to be classically-explainable. This approach can be motivated by a methodological version of Leibniz's principle~\cite{Leibniz}, by its equivalence to the existence of some positive quasiprobabilistic representation~\cite{negativity,Schmid2024structuretheorem}, and by its equivalence to the notion of classical-explainability arising in the framework of generalized probabilistic theories~\cite{SchmidGPT,Schmid2024structuretheorem}. By characterizing phenomena that cannot be reproduced in any such model, one learns what features of quantum theory are genuinely difficult to explain, and consequently what features are likely to provide quantum advantages in information processing and other physical protocols. Such phenomena have been found in the study of quantum computation~\cite{Schmid2022Stabilizer,shahandeh2021quantum}, state discrimination~\cite{schmid2018contextual,flatt2021contextual,mukherjee2021discriminating,Shin2021}, interference~\cite{Catani2023whyinterference,catani2022reply,catani2023aspects,giordani2023experimental}, compatibility~\cite{d2020classicality,selby2023incompatibility,selby2023accessible,PhysRevA.109.022239,PhysRevResearch.2.013011}, uncertainty relations~\cite{catani2022nonclassical}, metrology~\cite{contextmetrology}, thermodynamics~\cite{contextmetrology,comar2024contextuality,lostaglio2018}, weak values~\cite{AWV, KLP19}, coherence~\cite{rossi2023contextuality,Wagner2024coherence, wagner2024inequalities}, entanglement and the composition postulate~\cite{d2020classical,d2020classicality,erba2024compositionrulequantumsystems}, quantum Darwinism~\cite{baldijao2021noncontextuality}, information processing and communication~\cite{POM,RAC,RAC2,Saha_2019,Yadavalli2020,PhysRevLett.119.220402,fonseca2024robustness}, cloning~\cite{cloningcontext}, broadcasting~\cite{jokinen2024nobroadcasting}, pre- and post-selection paradoxes~\cite{PP1}, randomness certification~\cite{Roch2021}, psi-epistemicity~\cite{Leifer}, and Bell~\cite{Wright2023invertible,schmid2020unscrambling} and Kochen-Specker scenarios~\cite{operationalks,kunjwal2018from,Kunjwal16,Kunjwal19,Kunjwal20,specker,Gonda2018almostquantum}. 

Recently, the study of ontological models~\cite{harrigan2010einstein} and of generalized noncontextuality~\cite{spekkens2005contextuality} was reformulated~\cite{SchmidGPT,Schmid2024structuretheorem,ShahandehGPT,muller2023testing} in the language of generalized probabilistic theories~\cite{hardy,GPT_Barrett,janotta2013generalized} and in the newly developed framework of causal-inferential theories \cite{schmid2020unscrambling}. A GPT (as presented, e.g., in~\cite{Schmid2024structuretheorem}) is best thought of as a quotiented operational theory, where transformations constitute equivalence classes of operational procedures which cannot be distinguished by any experiment that could be performed within the theory. As was shown in Ref.~\cite{Schmid2024structuretheorem}, it follows that an ontological model of a GPT cannot be said to be either contextual or noncontextual, as a quotiented operational theory simply does not have any contexts on which one could ask whether or not a representation depends. Still, it was also shown in Ref.~\cite{Schmid2024structuretheorem} that there is a one-to-one correspondence between ontological models for a GPT and {\em noncontextual} ontological models for the operational theory that quotients to give that GPT.

A closely related framework is that of operational probabilistic theories, or OPTs~\cite{d2017quantum,chiribella2010probabilistic,chiribella2011}, and so it is natural to ask how to formulate ontological models and generalized noncontextuality also within this framework.

The framework of OPTs allows for both quotiented and unquotiented theories. Roughly speaking, an unquotiented OPT is analogous to an operational theory in the sense of Harrigan-Spekkens~\cite{harrigan2010einstein}, while a quotiented OPT is analogous to a GPT. However, quotiented OPTs have a key difference when compared to GPTs (at least as they are often presented, e.g., in~\cite{Schmid2024structuretheorem}), as the former includes an explicit formal structure capturing how certain collections of transformations form instruments. In the language of quantum theory, this would be a specification of which collections of completely-positive trace-nonincreasing maps form valid quantum instruments.\footnote{In a GPT, at least as presented in Ref.~\cite{Schmid2024structuretheorem}, one rather treats such a collection of transformations as a single process with a classical input and/or output that ranges over the elements of the collection.} One then intuitively expects that this extra structure enables the possibility of contextual representations. After all, the same transformation can appear in different instruments, and so this ``which instrument'' information constitutes a context upon which a representation could seemingly depend. Contrary to this intuition, we demonstrate that this is not the case for (quotiented) OPTs where the choice of a future instrument can be conditioned on a past measurement outcome (these are called \emph{strongly causal} OPTs~\cite{Rolino_2025}). What we show, in fact, is that an ontological model of a strongly causal quotiented OPT is entirely represented by its action on the transformations, and that this lifts to a representation of instruments in the obvious way (where the action on an instrument is given simply by its action on each of the individual transformations comprising it). Consequently, even if one takes into account the OPT instrument structure, ontological representations of strongly causal quotiented OPTs are found to be always noncontextual with respect to this structure.

An important consequence of this is that one can define noncontextuality of ontological models of {\em unquotiented} OPTs in terms of whether they factorize through the quotiented version of the OPT (analogous to what was shown for GPTs in Ref.~\cite{Schmid2024structuretheorem} and for causal-inferential theories in Ref.~\cite{schmid2020unscrambling}). 

This paper first provides a brief introduction to the OPT framework, illustrating how quotiented OPTs result in a linear structure and the conventional summation rule for coarse-graining. We subsequently introduce the notion of strong causality, which plays a key role for the following sections, and we conclude this introduction with a concise review of strictly classical OPTs.

In the subsequent section, we begin by introducing the notion of a general map between two arbitrary OPTs, with a particular focus on how such maps must accommodate the structure of instruments (or more generally, tests). We then define various properties that such maps can have, including outcome preservation, diagram preservation, coarse-graining preservation, and conditioning preservation. With these constraints in place, we define ontological models as maps from strongly causal OPTs to classical OPTs, that satisfy all the aforementioned properties. We then present a specific class of ontological models, namely noncontextual ontological models, and discuss why---unlike in the GPT framework---it might appear possible to have contextual representations for quotiented OPTs. Finally, we prove that this seeming possibility is illusory, and we demonstrate that the context induced by the structure of instruments in strongly causal quotiented OPTs can be disregarded without loss of generality. This allows ontological models to be defined solely at the level of transformations, ultimately establishing a one-to-one correspondence between noncontextual ontological models of strongly causal unquotiented OPTs and ontological models of strongly causal quotiented OPTs.

\section{Introduction to operational probabilistic theories}

\subsection{Primitives and compositional structure}
The primitive notions in an operational probabilistic theory are systems, events, tests, and probabilities. A system refers to physical objects that can be investigated in a laboratory, such as elementary particles, atoms, or radiation fields. An event describes some evolution taking one physical system to another, while a test is a collection of events with an outcome space representing possible evolutions that might occur in some probabilistic process; finally, an operational probabilistic theory must also specify how to combine these primitive elements to account for the probabilities of experimental outcomes.

We use \textbf{Sys($\Theta$)}, \textbf{Event($\Theta$)}, \textbf{Test($\Theta$)} to denote respectively the classes of systems, events, and tests in an OPT, $\Theta$. Systems will be denoted by uppercase Latin characters A, B, C, $\ldots$ $\in$ \textbf{Sys($\Theta$)}. The class of events from a system A to a system B will be denoted by  \textbf{Event}($\text{A}\rightarrow\text{B} $). Events t $\in$ \textbf{Event}($\text{A}\rightarrow\text{B} $) are represented as wired boxes, with the input-output direction represented as going from left to right:
\begin{equation}
\input{Diagramss/36_event_t.tikz}\;.
\end{equation}
A test with an outcome space X is denoted by $\text{T}_{\text{X}}^{\text{A}\rightarrow \text{B}}\in $  \textbf{Test}($\text{A}\rightarrow\text{B} )$. Each test  $\text{T}_{\text{X}}^{\text{A}\rightarrow \text{B}} $ is an indexed family\footnote{Note that another  approach in the OPT framework is to treat tests as multisets, which can be particularly useful when discussing the composition rules for tests. However, in this paper, we adopt the convention of defining tests as indexed families, as this choice proves more useful for formalising conditional instruments and ensuring they are well defined in the later sections. It is important to note, though, that using indexed families introduces certain subtleties in how test composition is handled. Conversely, while multisets allow for a straightforward definition of composition, they can introduce complications in the formal treatment of conditional instruments. For more detailed discussion, see the forthcoming work~\cite{erba2025categorical}.} of events from A to B, namely, 
\begin{equation}
\text{T}_{\text{X}}^{\text{A}\rightarrow \text{B}}=[ \text{t}_{x} ]_{\text{x}\in \text{X}}; \;\;\;\forall \text{x}\in \text{X},\;\;\text{t}_{\text{x}}\in\textbf{Event}(\text{A}\rightarrow \text{B}) .
\end{equation} 
Such a test is represented pictorially as
\begin{equation}
\input{Diagramss/36_test.tikz}\;.
\end{equation}
When clear from context, we may denote the test $\text{T}_{\text{X}}^{\text{A}\rightarrow \text{B}}$ as  $\text{T}_{\text{X}}$, or simply as T. A test with a single outcome will be called a \textit{singleton test}, and will have an outcome set denoted by $\star :=\{ \ast\}$. Operational probabilistic theories also incorporate sequential composition. Specifically, it is possible to arrange two tests in succession such that the outcome system of the first test serves as the input system for the second test (provided these systems are of the same type). That is, for all A, B, C $\in$ \textbf{Sys}($\Theta $) and all $\text{t}_{1}\in$ \textbf{Event}($\text{A}\rightarrow\text{B} $), $\text{t}_{2}\in$ \textbf{Event}($\text{B}\rightarrow\text{C} $), there exists an associative map $\circ$, called \textit{sequential composition}, and an event $\text{t}_{2}\circ\text{t}_{1}\in$ \textbf{Events}($\text{A}\rightarrow\text{C} $), such that the following sequential composition rule is defined for the two tests $\text{T}_{\text{X}}^{\text{A}\rightarrow\text{B}}$, $\text{T}_{\text{Y}}^{'\text{B}\rightarrow\text{C}}$:
\begin{equation}
(\text{T}^{'}\circ\text{T})_{\text{X}\times\text{Y}}^{\text{A}\rightarrow\text{C}}\equiv \text{T}^{'\text{B}\rightarrow \text{C}}_{\text{Y}}\circ \text{T}_{\text{X}}^{\text{A}\rightarrow \text{B}}:=[\text{t}^{'}_{y}\circ\text{t}_{x}]_{(\text{x},\text{y})\in\text{X}\times\text{Y}} \;.
\end{equation}
Such a sequential composition of two events is depicted as 
\begin{equation}
\input{Diagramss/37_seqcomp1.tikz}=\input{Diagramss/37_seqcomp.tikz}\; .
\end{equation}

To allow the OPT framework to encompass the possibility of doing nothing on systems we introduce an \textit{identity} event, $I_{\text{S}}$, for each system $\text{S}\in\textbf{Sys}(\Theta)$, as well as the (singleton) identity test, $\textbf{I}_{\star}^{\text{S}\rightarrow\text{S}}:=[I_{\text{S}}]$. Since the identity event represents doing nothing on the systems, it can be depicted as
\begin{equation}
\input{Diagramss/38_id.tikz}\;.
\end{equation} 
This, in particular, implies that the identity events satisfy 
\begin{equation}
I_{\text{B}}\circ\text{t}_{\text{x}}=\text{t}_{\text{x}}=\text{t}_{\text{x}}\circ I_{\text{A}},\;\;\;\forall\;\text{T}_{\text{X}},\;\;\forall\;\text{t}_{\text{x}}\in\text{T}_{\text{X}}^{\text{A}\rightarrow\text{B}},
\end{equation}
for all $\text{A}$ and $\text{B}$.

One can put two systems A and B in parallel to make the composite system AB, depicted as
\begin{equation}
\input{Diagramss/38_id1.tikz}\;.
\end{equation}
We define the \textit{parallel composition} map to combine two individual events into a single composite event, where the input (output) system of the composite event is given by the parallel composition of the inputs (outputs) of each individual event. Specifically, for two arbitrary events, $\text{t}_{1} \in$ \textbf{Event}($\text{A} \rightarrow \text{B}$) and $\text{t}_{2} \in$ \textbf{Event}($\text{C} \rightarrow \text{D}$), we define the parallel composition map $\boxtimes$ as an associative map that produces the composite event $\text{t}_{1} \boxtimes \text{t}_{2} \in$ \textbf{Event}($\text{AC} \rightarrow \text{BD}$). One can use the parallel composition map to define the parallel composition of tests via
\begin{equation}
(\text{T}\boxtimes\text{T}^{'})_{\text{X}\times\text{Y}}^{\text{AC}\rightarrow\text{BD}}\equiv\text{T}_{\text{X}}^{\text{A}\rightarrow\text{B}}\boxtimes\text{T}^{'\text{C}\rightarrow \text{D}}_{\text{Y}}:=[\text{t}_{x}\boxtimes \text{t}^{'}_{y}]_{(\text{x},\text{y})\in\text{X}\times\text{Y}}\;.
\end{equation}
This map is pictorially represented as 
\begin{equation}
\input{Diagramss/39_parcomp.tikz}=\input{Diagramss/39_parcomp1.tikz}=\input{Diagramss/39_parcomp2.tikz}\;\;.
\end{equation}
The $\circ$ and $\boxtimes$ maps must jointly satisfy:
\begin{equation}
\left(\text{t}_{1}\boxtimes \text{t}_{3}\right)\circ\left(\text{t}_{2}\boxtimes \text{t}_{4}\right)=\left(\text{t}_{1}\circ \text{t}_{2}\right)\boxtimes\left(\text{t}_{3}\circ\text{t}_{4}\right),
\end{equation}
for all systems A, B, C, D, E, F $\in$ \textbf{Sys}($\Theta $) and all $\text{t}_{1}\in$ \textbf{Events}($\text{A}\rightarrow\text{B} $), $\text{t}_{2}\in$ \textbf{Events}($\text{B}\rightarrow\text{C} $), $\text{t}_{3}\in$ \textbf{Events}($\text{D}\rightarrow\text{E} $), $\text{t}_{4}\in$ \textbf{Events}($\text{E}\rightarrow\text{F} $).\\

One can consider physical processes without inputs or outcomes, corresponding to situations where either the input or output is omitted from the physical description. To address this, one requires the notion of a trivial system I $\in$ \textbf{Sys}($\Theta$), which satisfies IA $=$ A $=$ AI for all A $\in$ \textbf{Sys}($\Theta$). Composing I with any system is equivalent to doing nothing, so to represent the trivial system, it is sufficient to leave blank spaces. Thus, events $\rho$ $\in$ \textbf{Events}($\text{I}\rightarrow\text{A}$), and a $\in$ \textbf{Events}($\text{A}\rightarrow \text{I}$) can be respectively depicted as
\begin{equation}
\input{Diagramss/40_preobs.tikz}\;.
\end{equation}
These are respectively called \textit{preparations} and \textit{observations}. 

To describe exchanging systems, an OPT requires the notion of \textit{braiding}, namely, a family  $\textbf{S}$ of invertible singleton tests, defined in a way that for any two systems A and B, the braiding $\textbf{S}$ contains two tests, $\textbf{S}_{\star}^{\text{AB}\rightarrow\text{BA}}=[s_{\text{A}, \text{B}}]$, and its inverse $(\textbf{S}_{\star}^{-1})^{\text{AB}\rightarrow\text{BA}}=[ s^{-1}_{\text{A}, \text{B}}]$ that are denoted as
\begin{equation}
\label{swap}
\input{Diagramss/41_swap1.tikz}=\input{Diagramss/41_swap.tikz}
\end{equation}
\begin{equation}
\input{Diagramss/41_swap3.tikz}=\input{Diagramss/41_swap2.tikz}\;\;.
\end{equation}
Braiding is required to satisfy the sliding property, namely,
\begin{equation}
\input{Diagramss/41_swap4.tikz}=\input{Diagramss/41_swap5.tikz}\;,
\end{equation}
for all A, B, C, D $\in$ \textbf{Sys}($\Theta$) and all $\text{t}_{1}\in$ \textbf{Event}($\text{A}\rightarrow\text{B} $), $\text{t}_{2}\in$ \textbf{Event}($\text{C}\rightarrow\text{D} $). In the case where $\text{s}^{-1}_{\text{A},\text{B}}= \text{s}_{\text{B},\text{A}}$, the OPT is said to be symmetric.

An important aspect of the formalism of OPTs, is that tests rather than events are treated as primitive -- the structure of events can be derived from that of tests but not vice versa. This is critical to the notion of ontological model that we introduce here, as we will define these as maps which act on tests, in contrast to earlier works (such as Ref.~\cite{Schmid2024structuretheorem}) which defined them as maps which act on events. 

\subsection{Probabilistic structure, coarse-graining, quotienting, and linear properties}
Tests with both trivial input and output are probability distributions $\text{P}_{\text{X}}=[ p_{x}]\in$ \textbf{Tests}($\text{I}\rightarrow\text{I} $) over the outcome space, X. That is, \textbf{Events}($\text{I}\rightarrow\text{I}$) $= [0,1]$, and for any test $\text{P}_\text{X}$ we have $\sum_{x\in\text{X}}p_x=1$. We denote probability distributions pictorially as
\begin{equation} \input{Diagramss/62_prob.tikz}\;,
\quad \text{and}\quad
\input{Diagramss/62_prob1.tikz}\;,    
\end{equation}
but, often, we will drop the diamond shape around these for clarity.
Moreover, parallel composition is given simply by multiplication of real numbers
\begin{equation}
p\boxtimes q:=pq,\;\;\;\forall p,q\in \textbf{Events}(\text{I}\rightarrow \text{I}).
\end{equation}
In particular, this means that given some arbitrary preparation $\rho_{\text{x}}$ for system A, followed by an event $\text{t}_{y}$ from system A to system B, and eventually followed by an observation $\text{a}_{\text{z}}$, the theory provides a joint probability distribution to the scenario:
  \begin{equation}
  \label{probability}
\input{Diagramss/43_joinp.tikz}:=p\left(x,y,z\vert\rho_{\text{X}}, \text{T}_{\text{Y}}, \text{A}_{\text{Z}}\right).
  \end{equation}

For all $p\in$ \textbf{Events}($\text{I}\rightarrow\text{I} $) and all $\text{t}\in\textbf{Events}(\text{A}\rightarrow\text{B})$ one can define \textit{multiplication by scalars} via 
\begin{equation}
\label{scmul}
p\cdot\input{Diagramss/36_event_t.tikz}:=\input{Diagramss/42_scaler.tikz}=\input{Diagramss/42_scaler1.tikz}\;.
\end{equation}

It is possible for an agent to perform a test conflating the outcomes within an arbitrary subset of the outcome space of a test. In other words, one may ignore the individual occurrences within a subset of events. To incorporate this possibility into the OPT framework, we adopt the notation and use the results of~\cite{erba2025categorical}. Let $\mathcal{K}\left(\text{X}\right)=\lbrace\text{X}_{k}\rbrace_{k\in K}$ denote an arbitrary partition of X, that is $\bigcup_{k\in\text{K}}\text{X}_{k}=\text{X}$, and for any $k_{1}\neq k_{2}$ one has $\text{X}_{k_{1}}\cap\text{X}_{k_{2}}=\emptyset$. Additionally, let $\textbf{Part}(\text{X})$ denote the set of all such partitions of $\text{X}$. Then, for all tests $\text{T}_{\text{X}}^{\text{A}\rightarrow\text{B}}$ and all partitions $\mathcal{K}\left(\text{X}\right)\in\textbf{Part}(\text{X})$, there exists a \textit{coarse-graining map} $\mathfrak{C}_{\mathcal{K}\left(\text{X}\right)}$ defined by the action
\begin{equation}
\mathfrak{C}_{\mathcal{K}(\text{X})}\left(\text{T}_{\text{X}}^{\text{A}\rightarrow\text{B}}\right):=[\text{t}_{\text{X}_{k}}]_{k\in\text{K}}\in\textbf{Test}(\text{A}\rightarrow\text{B}),
\end{equation}
such that when an element of an arbitrary partition $\text{X}_{k}=\lbrace x_{k}\rbrace$  is a singleton set for some $ x_{k}\in\text{X}$, the coarse-grained event associated with $\text{X}_{k}$,  is identical to the original event corresponding to the outcome $x_{k}$, so
\begin{equation}
\label{singcg}
\text{t}_{\lbrace x_{k}\rbrace}:=\text{t}_{x_{k}},\;\forall x_{k}\in \text{X}, \forall\mathcal{K}(\text{X})\in\textbf{Part}(\text{X}):\lbrace x_{k}\rbrace\in\mathcal{K}(\text{X}).
\end{equation}
The coarse-graining map must preserve both sequential and parallel composition of tests. Specifically, for an arbitrary pair of partitions $\mathcal{K}(\text{X}) = \lbrace\text{X}_k\rbrace_{k \in \text{K}}$ and $\mathcal{L}(\text{Y}) = \lbrace \text{Y}_{l}\rbrace_{l \in \text{L}}$, and for all $k \in \text{K}$, $l \in \text{L}$, one can define
\begin{equation}
\label{cgcomp1}
\left(\text{X}\times\text{Y}\right)_{\left(k,l\right)}:=\lbrace\left(x,y\right)\in \text{X}\times\text{Y}\vert x\in \text{X}_{k},y\in\text{Y}_{l}\rbrace.
\end{equation}
This definition can be used to define the compositionally induced partition of $\text{X} \times \text{Y}$ as
\begin{equation}
\label{cgcomp2}
\mathcal{KL}\left(\text{X}\times\text{Y}\right):=\lbrace \left(\text{X}\times\text{Y}\right)_{(k,l)}\rbrace_{(k,l)\in\text{K}\times\text{L}}\;.
\end{equation}
Having defined the compositionally induced partition, the preservation of both sequential and parallel composition of arbitrary tests $\text{A}_{\text{X}}$, $\text{B}_{\text{Y}}$, and $\text{C}_{\text{Z}}$ under the partitions $\mathcal{K}(\text{X})$, $\mathcal{L}(\text{Y})$, and $\mathcal{M}(\text{Z})$ requires that
\begin{align}
\label{cgcomp3}
\mathfrak{C}_{\mathcal{KL}\left(\text{X}\times\text{Y}\right)}\left(\text{B}_{\text{Y}}\circ\text{A}_{\text{X}}\right)=\mathfrak{C}_{\mathcal{L}(\text{Y})}\left(\text{B}_{\text{Y}}\right)\circ\mathfrak{C}_{\mathcal{K}(\text{X})}\left(\text{A}_{\text{X}}\right),\nonumber\\
\mathfrak{C}_{\mathcal{KM}\left(\text{X}\times\text{Z}\right)}\left(\text{A}_{\text{X}}\boxtimes\text{C}_{\text{Z}}\right)=\mathfrak{C}_{\mathcal{K}(\text{X})}\left(\text{A}_{\text{X}}\right)\boxtimes\mathfrak{C}_{\mathcal{M}(\text{Z})}\left(\text{C}_{\text{Z}}\right).
\end{align}
These conditions ensure that the coarse-graining map behaves consistently with the compositional structure of the OPT framework. The coarse-graining map also induces a binary operation $\mathlarger{\mathlarger{\mathop{\curlyvee}}}$ between distinct events of the associated test (corresponding to different outcomes), namely,
\begin{equation}
\label{cgevent}
\text{t}_{x}\curlyvee\text{t}_{\tilde{x}}:=\text{t}_{\lbrace x\rbrace\cup {\lbrace \tilde{x}\rbrace}},\;\; \forall \text{t}_{x},\text{t}_{\tilde{x}}\in \text{T}_{\text{X}}^{\text{A}\rightarrow\text{B}},  x\neq \bar{x}. 
\end{equation}
The above operation $\mathlarger{\mathlarger{\mathop{\curlyvee}}}$ is both associative and commutative, and both sequential and parallel composition distribute over $\mathlarger{\mathlarger{\mathop{\curlyvee}}}$. The $\mathlarger{\mathlarger{\mathop{\curlyvee}}}$ operation corresponds to the coarse-graining that $\text{X}_{k}$ in the partition $\mathcal{K}(\text{X})$ contains two elements, and the event associated with $\text{t}_{\text{X}_{k}}$ is given by (\ref{cgevent}). When $\text{X}_{k}$ contains more than two elements, the $\mathlarger{\mathlarger{\mathop{\curlyvee}}}$ operation can be naturally extended. That is, for an arbitrary partition $\mathcal{K}(\text{X}) = \lbrace\text{X}_k\rbrace_{k \in \text{K}}$, the corresponding coarse-grained event associated with $\text{X}_{k}$, with the outcomes of $\text{X}_{k}$ being labeled as $ 1 \leq x \leq m $, is given by
\begin{equation}
\bigcurlyvee_{x\in X_k} \text{t}_x:=\text{t}_{{\text{X}}_{k}}=\text{t}_{1}\curlyvee \text{t}_{2}\curlyvee\ldots\curlyvee \text{t}_{m}\;.
\end{equation}
Note that a coarse-grained event only depends on the events that are merged together, thereby it is in fact independent from the coarse-graining context, namely, from the specific partition or outcomes involved in the coarse-graining.

At this point, we have only relied on operational considerations to define the coarse-graining map. Later, we will see that by considering scalar tests as probability distributions and using them to construct quotiented OPTs, a linear structure will be induced, which can be used to derive the standard summation rule for the coarse-graining map.

Having equipped the OPT framework with probabilities as in (\ref{probability}), it turns out that for scalar tests, i.e., probability distributions, the coarse-graining of events is given by the usual sum of real numbers. Namely, for all $\text{P}_{\text{X}}^{\text{I}\rightarrow\text{I}}=[ p_{x} ] _{x \in \text{X}} \in \textbf{Test}(\text{I}\rightarrow\text{I})$, and for every partition $\mathcal{K}(\text{X}) = \lbrace\text{X}_k\rbrace_{k \in \text{K}}\in\textbf{Part}(\text{X})$, one has
\begin{equation}
\label{cgprob}
\mathlarger{\mathlarger{\mathop{\curlyvee}}}_{x \in \text{X}_{k}} p_{x}=p_{\text{X}_{k}}=\sum_{x \in \text{X}_{k}}p_{x}.
\end{equation}
See \cite{erba2025categorical} for details.

Having a rule for assigning joint probabilities to events allows us to define an equivalence relation. In an OPT $\Theta$, two arbitrary events $\text{t}_{1},\text{t}_{2}\in\textbf{Events}(\text{A}\to\text{B})$, are said to be \emph{operationally equivalent}, denoted as $\text{t}_{1}\sim\text{t}_{2}$, if one has
\begin{equation}
\label{eqQ}
\input{Diagramss/44_eq.tikz}=\input{Diagramss/44_eq1.tikz}\;,
\end{equation}
For all E $\in\textbf{Sys}(\Theta)$, for all $\rho\in\textbf{Events}(\text{I}\rightarrow\text{AE})$, and for all $\text{a}\in\textbf{Events}(\text{BE}\rightarrow\text{I})$.

This allows us to define a \textit{quotiented OPT} by defining \textit{equivalence classes} of events:
\begin{equation}
\label{eqQ1}
\textbf{Transf}\left(\text{A}\rightarrow\text{B}\right):=\textbf{Events}\left(\text{A}\rightarrow\text{B}\right)/\sim\;.
\end{equation} 
These equivalence classes of events are called \textit{transformations} from system A to system B. Equivalence classes of preparations \textbf{St}(A) $:=\textbf{Transf}(\text{I}\rightarrow\text{A})$ are called the states of A, while the equivalence class of observations $\textbf{Eff}(A):=\textbf{Transf}(\text{A}\rightarrow
\text{I})$ are called the effects of A. In this regard, two tests $\text{T}^{\text{A}\rightarrow\text{B}}_{\text{X}}=[\text{t}_{x}]_{x\in\text{X}}$ and $\text{T}^{'\text{A}\rightarrow\text{B}}_{\text{X}}=[\text{t}^{'}_{x}]_{x\in\text{X}}$, for an arbitrary outcome space X, and for arbitrary systems A and B, are equivalent if for all $x\in\text{X}$, one has $\text{t}_{x}\sim\text{t}^{'}_{x}$. The equivalent classes of tests are called \textit{instruments} and are represented as $\textbf{Instr}(\text{A}\to \text{B})$. Sometimes, similar to tests, we denote the instrument $\text{T}_{\text{X}}^{\text{A}\rightarrow \text{B}}$ as T. States of an arbitrary system A are denoted as $\vert\rho)_{\text{A}}$, while effects are denoted as $(a\vert_{A}$. The parallel composition between states and effects is represented as $\vert\rho_{1})_{\text{A}_{1}}\vert\rho_{1})_{\text{A}_{1}}$ and  $(a_{1}\vert_{A_{1}}(a_{2}\vert_{A_{2}}$, and the sequential composition between a state and an effect is denoted as $(a\vert \rho)_{A}\in\left[0, 1\right]$. 

Starting from an arbitrary OPT $\Theta$, one can define a new OPT, $\tilde{\Theta}$, which has the same systems as in $\Theta$. In $\tilde{\Theta}$, the events and tests correspond to the transformations and instruments from $\Theta$. The identity in $\tilde{\Theta}$ for system $S$, denoted $\mathcal{I}_S$, is the equivalence class containing the identity $I_S$ from $\Theta$. Each event $t$ in $\Theta$ is represented in $\tilde{\Theta}$ by its equivalence class $\tilde{t}$. It follows directly that this equivalence relation is a congruence, and it is compatible with sequential and parallel composition, as well as coarse-graining~\cite{Schmid2024structuretheorem,erba2025categorical}. Consequently, the resulting quotiented OPT $\tilde{\Theta}$ is itself a valid OPT. From this point forward, we refer to $\Theta$ as an \textit{unquotiented} OPT and to $\tilde{\Theta}$ as a \textit{quotiented} OPT. Sometimes, when clear from the context, we use $\Theta$ (and generally letters without the tilde symbol) to refer to both unquotiented and quotiented OPTs.

Furthermore, considering (\ref{eqQ}), one has two important consequences. First, \textbf{St}(A) is a set of functionals from \textbf{Eff}(A) to the real interval $\left[0, 1\right]$ and vice versa. Second, the set of states is separating for the set of effects, and vice versa. That is, for every pair of states $\vert\rho)_{\text{A}}$ and $\vert\sigma)_{\text{A}}\in \textbf{St}(\text{A})$ that are distinct, $\vert\rho)_{\text{A}}\neq\vert\sigma)_{\text{A}}$, there exists an effect $(a\vert_{\text{A}}\in\textbf{Eff}(A)$ such that $(a\vert \rho)_{A}\neq(a\vert \sigma)_{A}$ (and similarly for every pair of different effects). These two consequences induce a linear structure; that is, the coupling between states and effects and the fact that they are separating for each other, can be used to define a complete class of linearly independent vectors in \textbf{St}(A) spanning a real vector space $\textbf{St}_{\mathbb{R}}(\text{A}):=\text{Span}_{\mathbb{R}}\textbf{St}(\text{A})$. Similarly, since \textbf{Eff}(A) is a class of non-negative linear functionals on \textbf{St}(A), it spans the dual space $\textbf{Eff}_{\mathbb{R}}(\text{A})=\text{span}_{\mathbb{R}}\textbf{Eff(A)}=\textbf{St}_{\mathbb{R}}(\text{A})^{*}$. The dimension $D_{\text{A}}:=\text{dim}\;\textbf{St}_{\mathbb{R}}$ will be called the dimension of system $A$. $D_{\text{A}}$ is the minimum number of probabilities that one must know to determine the state of a system A. In this work, we only consider finite-dimensional quotiented OPTs, for which one has $\text{dim}\;\textbf{St}_{\mathbb{R}}(\text{A})=\text{dim}\;\textbf{Eff}_{\mathbb{R}} (\text{A})$ for all systems A. 

Like with states and effects, transformations can also be viewed as living in vector spaces. That is, equation~(\ref{eqQ}) allows us to view transformations from $\text{A}$ to $\text{B}$ as functionals from state-effect pairs to probabilities, which similarly to the case of states, enables the assignment of a vector space spanned by transformations, namely $\textbf{Transf}_{\mathbb{R}}(\text{A}\rightarrow\text{B}):=\text{Span}_{\mathbb{R}}\;\textbf{Transf}(\text{A}\rightarrow\text{B})$.

Also as a consequence of~(\ref{eqQ}), every transformation $\text{t}\in \textbf{Transf}(\text{A} \rightarrow \text{B})$ naturally defines a linear map $\hat{\text{t}}_{\text{E}}$ from the vector space $\textbf{St}_\mathbb{R}(\text{AE})$ to $\textbf{St}_\mathbb{R}(\text{BE})$ for any system E. Every transformation 
t is therefore  characterized by the family of linear maps $\lbrace \hat{\text{t}}_{\text{E}}\rbrace_{ \text{E} \in \text{Sys}(\Theta)}$. That is, for any two transformations $\text{t}_{1}$ and $\text{t}_{2}$ $\in \textbf{Transf}(\text{A} \rightarrow\text{B}) $, one has
\begin{equation}
\text{t}_1=\text{t}_2\iff {\hat{\text{t}_1}}_{\text{E}}=\hat{\text{t}_2}_{\text{E}},\;\forall\text{E}\in\textbf{Sys}(\Theta).
\end{equation}
Note that, in general, the entire infinite family of linear maps $\lbrace \hat{\text{t}}_{\text{E}}\rbrace_{ \text{E} \in \text{Sys}(\Theta)}$, for all systems $\text{E}$, may be required to fully specify a transformation. However, under certain conditions, this may not be necessary. For example, in ``locally tomographic'' theories each transformation is fully determined by the single map associated with the local action, where the ancilla $\text{E}$ is trivial. 

Due to the compositional properties of the coarse-graining map (\ref{cgcomp1}-\ref{cgcomp3}), in any unquotiented OPT $\Theta$, one has
\begin{equation} 
\label{cgquo}
\input{Diagramss/56_cg4.tikz} = \mathlarger{\mathlarger{\mathop{\curlyvee}}}_{x \in \text{X}_{k}} \input{Diagramss/56_cg5.tikz}\;,
\end{equation}
for all systems A and B $\in\textbf{Sys}(\Theta)$, all tests $\text{T}_{\text{X}}^{\text{A}\rightarrow\text{B}}=[\text{t}_{\text{x}} ]_{\text{x}\in\text{X}}\in \textbf{Test}(\Theta)$, all partitions $\mathcal{K}(\text{X})=\lbrace \text{X}_{k}\rbrace _{k\in\text{K}}\in \textbf{Part}(\text{X})$, all systems E $\in\textbf{Sys}(\Theta)$, all $\rho\in\textbf{Events}(\text{I}\rightarrow\text{AE})$, and all $\text{a}\in\textbf{Events}(\text{BE}\rightarrow\text{I})$, for the coarse-grained event $\text{t}_{\text{X}_{k}}$. To see this, consider an arbitrary test T with the outcome space X, and without loss of generality, assume $\rho$ and $a$ are singleton preparation and observation. Then, the outcome space for the scalar of the form $a_{\text{BE}}\circ \left(T_{\text{X}}^{\text{A}\rightarrow\text{B}}\boxtimes\mathcal{I}_{\text{E}}\right)\circ\rho_{AE}$ would still be X. For an arbitrary partition $ \mathcal{K}(X) $, applying coarse-graining to the test T, is equivalent to applying the coarse-graining to the entire compositional scenario, that is the scalar, which leads to (\ref{cgquo}), and it can be straightforwardly generalized to arbitrary preparations and arbitrary observations that are not singleton. 

Now, using the induced linear structure in quotiented OPTs, and also the fact that one has the summation rule for the coarse-graining of probabilities (\ref{cgprob}), would result in the standard summation rule for coarse-grained transformations, as it is shown in Ref.~\cite{erba2025categorical}. To illustrate this in an arbitrary quotiented OPT $\Theta$,  consider an arbitrary instrument $\text{T}_{\text{X}}^{\text{A}\rightarrow\text{B}}=[\text{t}_{x}] _{x\in\text{X}}$ for some arbitrary systems A and B $\in\textbf{Sys}(\Theta)$. Then for all $\text{E}\in \textbf{Sys}(\Theta)$, $\vert \rho)_{\text{AE}}\in\textbf{St}(\text{AE})$, $(\text{a}\vert_{\text{BE}}\in \textbf{Eff}(\text{BE})$, and all choices of subsets of outcomes $\text{X}_{k}\subseteq \text{X}$, one has
\begin{align}
(a\vert_{\text{BE}}( \text{t}_{\text{X}_{k}}\boxtimes\mathcal{I}_{\text{E}})\vert
\rho)_{AE}=&\mathlarger{\mathlarger{\mathop{\curlyvee}}}_{x \in \text{X}_{k}} (a\vert_{\text{BE}}\left( \text{t}_{x}\boxtimes\mathcal{I}_{\text{E}}\right)\vert
\rho)_{AE}\nonumber\\
= \sum_{x \in \text{X}_{k}}(a\vert_{\text{BE}}\left( \text{t}_{x}\boxtimes\mathcal{I}_{\text{E}}\right)\vert \rho)_{AE}=&(a\vert_{\text{BE}}(\sum_{x \in \text{X}_{k}} \text{t}_{x}\boxtimes\mathcal{I}_{\text{E}})\vert\rho)_{AE}\;.
\end{align}
Where (\ref{cgquo}) and (\ref{cgprob}), together with the property of linearity, are used, which consequently implies
\begin{equation}
\mathlarger{\mathlarger{\mathop{\curlyvee}}}_{x \in \text{X}_{k}} \text{t}_{x}=\text{t}_{\text{X}_{k}}=\sum_{x \in \text{X}_{k}} \text{t}_{x}.
\end{equation} 
In order to accommodate zero probabilities into the framework, we assume that $0\in\textbf{Transf}(\text{I}\rightarrow\text{I})$. Scalar multiplication (\ref{scmul}) also implies the existence of \textit{null} transformations $\varepsilon_{\text{A}\rightarrow\text{B}}:=0\cdot \text{t}\in\textbf{Transf}(\text{A}\rightarrow\text{B})$ for any $\text{t}\in\textbf{Transf}(\text{A}\rightarrow\text{B})$ and every system A and B $\in\textbf{Sys}(\Theta)$, such that for every $\text{E}\in\textbf{Sys}(\Theta)$, $\vert\rho)_{\text{AE}}\in\textbf{St}(\text{AE})$, $(\text{a}\vert_{\text{BE}}\in\textbf{Eff}(\text{BE})$ one has $(\text{a}\vert_{\text{BE}}\left(\varepsilon_{\text{A}\rightarrow\text{B}}\boxtimes \mathcal{I}_{\text{E}}\right)\vert \rho)_{\text{AE}}=0$.  

A transformation $\text{t}\in\textbf{Transf}(\text{A}\rightarrow\text{B})$ is called \textit{deterministic} if there exists a singleton instrument such that $\text{T}_{\star}=[\text{t}]\in\textbf{Instr}(\text{A}\rightarrow\text{B})$. In other words, deterministic transformations, also known as \textit{channels}, are the transformations happening with certainty, i.e., with the marginal probability 1, and will be denoted by $\textbf{Transf}_{1}(\text{A}\rightarrow\text{B})$. For any deterministic transformation in a quotiented OPT, each element of the corresponding equivalence class in the associated unquotiented OPT is referred to as a \textit{deterministic event}. Deterministic states and deterministic effects are respectively denoted as $\textbf{St}_{1}(\text{A})$ and $\textbf{Eff}_{1}(\text{A})$. For any given instrument $\text{T}_{\text{X}}^{\text{A}\rightarrow\text{B}}$, the full coarse-graining of the instrument is a deterministic transformation i.e., $\sum_{x\in\text{X}}\text{t}_{x}\in\textbf{Transf}_{1}(\text{A}\rightarrow\text{B})$. 

\subsection{Causality and conditioning}\label{subsec:causality}
The assumption that the probability distribution of preparation instruments is independent of the choice of observation instrument at their output is often called \textit{causality}. This is equivalent to the \textit{uniqueness of deterministic effect} for every system~\cite{d2017quantum}. Namely, a quotiented OPT is causal if and only if there exists a unique deterministic effect for every system A, depicted as
\begin{equation}
\begin{tikzpicture}
	\begin{pgfonlayer}{nodelayer}
		\node [style=infupground] (20) at (0.75, 0) {};
		\node [style=none] (21) at (-0.75, 0) {};
		\node [style=none] (22) at (-0.75, 0) {};
		\node [style=none] (23) at (-0.25, 0.5) {\scriptsize$\text{A}$};
	\end{pgfonlayer}
	\begin{pgfonlayer}{edgelayer}
		\draw [qWire] (21.center) to (20);
	\end{pgfonlayer}
\end{tikzpicture}
\;.
\end{equation}
We denote the elements of the corresponding equivalence class of the deterministic effect in the associated unquotiented OPT as
\begin{equation}
\input{Diagramss/50_discard1.tikz}\;.
\end{equation}
where the letter \textit{c} represents an element within this equivalence class. It is well known that this notion of causality is equivalent to the principle of \emph{no-signalling from the future}~\cite{d2017quantum}, and it was also shown that it entails \emph{(spatial) no-signalling}---also known as \emph{no-signalling without interaction}~\cite{d2017quantum,coecke2014terminality}.

In the OPT framework, one can also consider \textit{conditional tests}, where the choice of which test to implement depends on the outcome of an earlier test. 
For example, given a test $\text{A}_{\text{X}}\in \textbf{Test}(\text{A}\to \text{B})$, the occurrence of an outcome $x\in \text{X}$ might dictate the implementation of a test $\text{B}_{\text{Y}^{(x)}}^{(x)}\in \textbf{Test}(\text{B}\rightarrow\text{C})$, which will have an outcome in $\text{Y}^{(x)}$. The space of outcomes for the conditional test will be given by 
\beq \bigsqcup_{x\in \text{X}}\text{Y}^{(x)} = \bigcup_{x\in \text{X}}\bigcup_{y^{(x)}\in \text{Y}^{(x)}}(x,y^{(x)}),
\eeq 
where each outcome $(x,y^{(x)})$ corresponds to the event $b^{(x)}_{y^{(x)}}\circ a_x$ having occurred. That is, the full conditional test is given by $\left[b^{(x)}_{y^{(x)}}\circ a_x\right]_{x\in \text{X}, y^{(x)}\in \text{Y}^{(x)}}$. 

We can, moreover, think about conditioning as a binary operation, denoted $\rhd$, which takes a test $\text{A}_\text{X}$ and an indexed family of tests $[\text{B}^{(x)}_{\text{Y}^{(x)}}]_{x\in \text{X}}$ as an input, such that 
\beq\label{condtest}
\left[\text{B}^{(x)}_{\text{Y}^{(x)}}\right]_{x\in X}\rhd \text{A}_\text{X}:=\left[b^{(x)}_{y^{(x)}}\circ a_x\right]_{x\in \text{X}, y^{(x)}\in \text{Y}^{(x)}}.
\eeq
Note that in general OPTs, this is only a partial operation, as it may be the case that $\left[b^{(x)}_{y^{(x)}}\circ a_x\right]_{x\in X, y^{(x)}\in Y^{(x)}}$ is not a valid test within the OPT.

With slight abuse of notation, we depict conditional tests as
    \begin{equation}
\label{Cond}
\input{Diagramss/53_condp5.tikz}:=\input{Diagramss/53_condp2.tikz}\;.
\end{equation}
An OPT is called \textit{strongly causal}~\cite{Rolino_2025} if $\rhd$ is a total operation, that is, if the conditional test defined as in (\ref{condtest}) is a valid test within the theory for all systems A, B, C $\in \textbf{Sys}(\Theta)$ and for all outcome sets $\text{X}$, all tests $\text{A}_{\text{X}}\in \textbf{Test}(\text{A}\rightarrow\text{B})$, and all $\text{X}$ indexed families of tests $\left[\text{B}^{(x)}_{\text{Y}^{(x)}}\right]_{x\in \text{X}}$ with $\text{B}^{(x)}_{\text{Y}^{(x)}}\in \textbf{Test}(\text{B}\rightarrow\text{C})$ for all $x\in \text{X}$. In other words, every possible conditional test must also be a valid test within the theory. Since we have defined conditional tests and the notion of strong causality for an arbitrary unquotiented OPT, they are also defined for quotiented OPTs as a specific case.

The justification for considering this property as a stronger notion of causality is presented in~\cite{chiribella2010probabilistic}, where it is argued that the possibility to implement any conditional instrument in an arbitrary quotiented OPT implies the uniqueness of the deterministic effect for every system. In other words, causality is a necessary condition to have conditional instruments, and a theory where every conditional instrument is possible is causal. Note that the converse is not true, as shown by Refs.~\cite{PhysRevA.109.022239,Rolino_2025}.

In OPTs where every arbitrary probability can be produced, strong causality also implies that every convex combination of tests is itself a test. Specifically, in any arbitrary strongly causal OPT, given a conditional test, based on an arbitrary probability distribution \( \text{P}_\text{X}=[p_{x} ]_{x \in \text{X}} \), and arbitrary tests \( \text{D}^{(x)}_{\text{Y}^{(x)}} = [ \text{d}_{y}^{(x)} ]_{y \in \text{Y}^{(x)}} \in \textbf{Test}(\text{A} \rightarrow \text{B}) \), the resulting conditional test
\beq
\left[\text{D}^{(x)}_{\text{Y}^{(x)}}\right]_{x\in \text{X}}\rhd (\text{P}_X\boxtimes \textbf{I}^{\text{A}\to \text{A}}_\star) = \left[p_x\cdot d^{(x)}_{y^{(x)}}\right]_{x\in \text{X}, y^{(x)}\in \text{Y}^{(x)}}\;,
\eeq
is itself a test in the theory. If we then fully coarse-grain the $\text{X}$ variable, then this can be interpreted as a convex combination of tests, demonstrating that every convex combination of tests is itself a test.

\subsection{Classical theory}
Classical theory (CT) is a strongly causal quotiented OPT $\Delta$ in which systems are labeled by finite sets $\Lambda$, composite systems are given by the Cartesian product, that is, $\Lambda\Gamma:=\Lambda\times\Gamma$, and the trivial system is given by the singleton set $\star$. 
Transformations from system $\Lambda$ to system $\Gamma$ are given by the substochastic matrices M $\in \textbf{SubStoch}(\Lambda\rightarrow\Gamma)$, and deterministic transformations by stochastic matrices  S $\in \textbf{Stoch}(\Lambda\rightarrow\Gamma)$. As a special case of this, states are (subnormalised) probability distributions, effects are response functions, and scalars are probabilities.

Sequential composition is the composition of (sub)stochastic matrices, and parallel composition is the usual tensor product. 

Finally, instruments in the classical theory are all collections $[\text{t}_{x}] _{x\in\text{X}}\subset \textbf{SubStoch}(\text{A}\rightarrow\text{B})$ such that $\sum_{x\in\text{X}} \text{t}_{x}\in\textbf{Stoch}(\text{A}\rightarrow\text{B})$, and will be denoted as $[ \text{t}_{x}] _{x\in\text{X}}\in\textbf{CInstr}(\text{A}\rightarrow\text{B})$~\cite{erba2025categorical}.
\section{Ontological models of OPTs} 
\subsection{Maps between different OPTs}
In the following subsections, we discuss the notion of ontological models and noncontextual ontological models for OPTs. We will first examine the general concept of a map between two OPTs, and then focus on a specific class of these maps that constitute ontological models. Since the domain and codomain of these maps can be either unquotiented or quotiented OPTs, we will for simplicity use the terms ``tests'' and ``events'' to also refer to instruments and transformations within quotiented OPTs (although this is nonstandard terminology). Given two arbitrary operational probabilistic theories, we can define the following map:
\begin{definition}(Maps between two OPTs). A map $\input{Diagramss/51_mapopt.tikz}$, where the domain and codomain of this map can be either quotiented or unquotiented OPTs, takes every system $\text{A} \in \textbf{Sys}(\Theta)$ to a system $\phi(\text{A}) \in \textbf{Sys}(\Phi)$ and maps every test $\text{T}_{\text{X}}^{\text{A} \rightarrow \text{B}}$ to a $\phi{\left(\text{T}_{\text{X}}^{\text{A} \rightarrow \text{B}}\right)} \in \textbf{Test}(\Phi)$ with outcome space $\phi{\left(\text{X}\,\vert \text{T}_{\text{X}}\right)}$. We depict this as
\begin{equation}
\phi::\input{Diagramss/36_test.tikz}\mapsto \input{Diagramss/51_mapopt2.tikz}\;.
\end{equation}
\end{definition}
If the map $\phi$ preserves the outcome space, such that $\phi\left(\text{X}\,\vert \text{T}_{\text{X}}\right)=\text{X}$, then for any arbitrary test $\text{T}_{\text{X}}^{\text{A}\rightarrow\text{B}}= [\text{t}_{x}]_{x\in X}$ and its corresponding events, we can use the following notation:
\begin{equation}
\scalebox{0.9}{$
\phi{\left(\text{T}_{\text{X}}^{\text{A}\rightarrow\text{B}}\right)}
:=  \left[\phi(\text{t}_{x}\vert \text{T}_{\text{X}}^{\text{A}\rightarrow\text{B}})\right]_{x\in X}\;. 
$}
\end{equation}
In this regard, the event $\phi(t_{x}\vert \text{T}_{\text{X}}^{\text{A}\rightarrow\text{B}})\in \phi{\left(\text{T}_{\text{X}}^{\text{A}\rightarrow\text{B}}\right)}$ can be viewed as the event that the map $\phi$ assigns to the event $\text{t}_{x}\in \text{T}_{\text{X}}^{\text{A}\rightarrow\text{B}}$. For the sake of brevity, we sometimes use the shorthand notation $\phi(\text{T})$ for $\phi{\left(\text{T}_{\text{X}}^{\text{A}\rightarrow\text{B}}\right)}$, and so write
\begin{equation}
\label{map}
\phi(\text{T}):= \left[ \phi(\text{t}_x\vert \text{T})\right]_{x\in X}.
\end{equation}
For a map $\phi$ that preserves the outcome space, we represent how the map transforms events pictorially as
\begin{equation}
\phi::\input{Diagramss/47_ontmodelunq1.tikz}\mapsto \input{Diagramss/51_mapopt1.tikz}\;.
\end{equation}

\subsubsection{Diagram preservation}
The specific class of maps that are relevant in this work are those with the property of \textit{diagram preservation} defined as follows:
\begin{definition}
\label{dpdef}
(Diagram-preserving maps). A diagram-preserving map $\input{Diagramss/52_dp.tikz}$ is one that, for an arbitrary test in $\textbf{Test}(\Theta)$ with a composite input system $\text{A}_{1} \boxtimes \text{A}_{2} \boxtimes \ldots \boxtimes \text{A}_{n}$ and composite output system $\text{B}_{1} \boxtimes \text{B}_{2} \boxtimes \ldots \boxtimes \text{B}_{m}$, assigns a corresponding test in $\textbf{Test}(\Phi)$ with composite input system $\phi_{\text{\tiny DP}}(\text{A}_{1}) \boxtimes \phi_{\text{\tiny DP}}(\text{A}_{2}) \boxtimes \ldots \boxtimes \phi_{\text{\tiny DP}}(\text{A}_{n})$ and composite output system $\phi_{\text{\tiny DP}}(\text{B}_{1}) \boxtimes \phi_{\text{\tiny DP}}(\text{B}_{2}) \boxtimes \ldots \boxtimes \phi_{\text{\tiny DP}}(\text{B}_{m})$
\begin{equation}
\phi_{DP}::\input{Diagramss/52_dp0.tikz}\mapsto\input{Diagramss/52_dp1.tikz}\;,
\end{equation}
Such that composing tests before or after applying the map remains equivalent. For example,
\begin{align}
&\input{Diagramss/52_dp2.tikz}\nonumber \\
&=\input{Diagramss/52_dp3.tikz}\;.
\end{align}
Diagram-preserving maps also preserve the identity 
\begin{equation}
    \input{Diagramss/52_dp8.tikz}= \begin{tikzpicture}
	\begin{pgfonlayer}{nodelayer}
		\node [style=none] (0) at (-2, 0) {};
		\node [style=none] (1) at (2, 0) {};
		\node [style=none] (2) at (0, 0.5) {\scalebox{0.82}{$\phi_{\text{\tiny DP}}(\scalebox{0.8}{$\text{A}$})$}
};
	\end{pgfonlayer}
	\begin{pgfonlayer}{edgelayer}
		\draw [qWire] (0.center) to (1.center);
	\end{pgfonlayer}
\end{tikzpicture}
\;,
\end{equation}
and, similarly,  preserve braidings. 
\end{definition}

As for general maps between two OPTs, when $\phi_{DP}$ satisfies outcome preservation, it is straightforward to define it at the level of events: 
\begin{equation}
\phi_{DP}::\input{Diagramss/52_dp4.tikz}\mapsto\input{Diagramss/52_dp5.tikz}\;,
\end{equation}
Thus, for composite tests, it can likewise be defined for the corresponding composite events, i.e.,
\begin{align} 
&\input{Diagramss/52_dp6.tikz}\nonumber\\
&=\scalebox{0.9}{\input{Diagramss/52_dp7.tikz}}\;.
\end{align}
\subsubsection{Coarse-graining preservation}
For maps with outcome space preservation, we can define the coarse-graining-preserving maps as follows:
\begin{definition}(Coarse-graining-preserving maps). A coarse-graining-preserving map $\input{Diagramss/58_cgp.tikz}$ is one that preserves the outcome space and commutes with the coarse-graining operation $\curlyvee$, that is
\begin{align}
&\input{Diagramss/58_cgp1.tikz}\nonumber\\
=\mathlarger{\mathlarger{\mathop{\curlyvee}}}_{x \in \text{X}_{k}} &\input{Diagramss/58_cgp2.tikz}\; .
\end{align}
\end{definition}
This means the map ensures that for any coarse-grained event, the result of applying the map to the coarse-grained event is the same as the coarse-graining of the map’s action on the individual events.
\subsubsection{Probability preservation}
\begin{definition}\label{ppdef}(Probability-preserving maps). A probability-preserving map $\input{Diagramss/59_pp.tikz}$ is an outcome space preserving map which leaves probability distributions invariant. Specifically, for every probability distribution $P_{X}=[ p_{x}]_{x\in\text{X}}$, one has
\begin{equation}
\input{Diagramss/59_pp1.tikz}=\input{Diagramss/59_pp4.tikz}\;.
\end{equation}
Or, equivalently,
\begin{equation}
\input{Diagramss/59_pp5.tikz}=\input{Diagramss/59_pp6.tikz}\;,
\end{equation}
for every $x \in \text{X}$.
\end{definition}

\subsubsection{Conditioning preservation}
  \begin{definition}
  \label{cpdef}(Conditioning-preserving maps), A conditioning-preserving map $\input{Diagramss/61_cp.tikz}$ is an outcome space preserving map such that for any conditional test
\begin{equation}
    \input{Diagramss/53_condp1.tikz}=\input{Diagramss/53_condp2.tikz}
    \end{equation}  
in $\Theta$, the corresponding test in $\Phi$ is the conditional test given by
  \begin{align}
&\hspace{0.6cm}\input{Diagramss/61_cp1.tikz}\nonumber\\
  =& \input{Diagramss/61_cp2.tikz}\;.
  \end{align}
  \end{definition}
\subsection{Ontological models for quotiented and unquotiented OPTs}
In this section, we examine specific types of maps between OPTs—namely, ontological models.

\begin{definition}\label{ontmodel}
(Ontological models of strongly causal OPTs). An ontological model of a strongly causal OPT is a diagram-preserving map $\input{Diagramss/47_ontmodelunq.tikz}$, where $\xi$ satisfies outcome space preservation, coarse-graining preservation, and conditioning preservation, and is depicted as
\begin{equation}
\xi::\input{Diagramss/47_ontmodelunq3.tikz}\rightarrow \input{Diagramss/47_ontmodelunq4.tikz}\;.
\end{equation}
Due to outcome preservation, $\xi$ can be represented at the level of events
\begin{equation}
\xi::\input{Diagramss/47_ontmodelunq1.tikz}\rightarrow \input{Diagramss/47_ontmodelunq2.tikz}\;.
\end{equation}
\end{definition}
An ontological model of a strongly causal OPT assigns each event to a substochastic process within CT. It also maps every preparation and observation event in an OPT to corresponding states and effects in classical theory, and associates each system A with a system in CT, where $\Lambda_{\text{A}}$ denotes the dimension of the associated classical system, i.e., the dimension of the ontic space \cite{harrigan2010einstein}. For an arbitrary test $\text{T} = \left[ \text{t}_{x}\right]_{x\in X}$ from system $\text{A}$ to system $\text{B}$ in any strongly causal OPT, and given that ontological models preserve the outcome space, the corresponding instrument in an ontological model, from system $\Lambda_{\text{A}}$ to system $\Lambda_{\text{B}}$, similar to (\ref{map}), can also be denoted as
\begin{equation}
\xi{\left( \text{T} \right)} =\left[\xi{\left(\text{t}_{x}\vert \text{T} \right)}\right]_{x\in X}.
\end{equation}
Since ontological models must satisfy outcome space preservation, they inherently preserve determinicity. That is, deterministic events are represented by deterministic transformations. This is because deterministic events correspond to singleton tests, and due to outcome space preservation, the resulting instruments in ontological models remain singleton instruments, and so deterministic. This, in turn, implies the following for the equivalence class of deterministic effects:
 \begin{equation}
    \input{Diagramss/54_discp.tikz}=\begin{tikzpicture}
	\begin{pgfonlayer}{nodelayer}
		\node [style=infupground] (20) at (0.75, 0) {};
		\node [style=none] (21) at (-0.75, 0) {};
		\node [style=none] (22) at (-0.75, 0) {};
		\node [style=none] (24) at (-0.25, 0.5) {\scriptsize$\Lambda_{\text{A}}$};
	\end{pgfonlayer}
	\begin{pgfonlayer}{edgelayer}
		\draw (21.center) to (20);
	\end{pgfonlayer}
\end{tikzpicture}
\;.
\end{equation}
An essential property that ontological models must satisfy is the preservation of probabilities, as defined in Definition~\ref{ppdef}. We do not include this assumption in Definition \ref{ontmodel} for two reasons. First, this property is not necessary in order to derive our main result. Second, as we prove in Appendices~\ref{appA} and~\ref{app:prob-pres}, probability preservation follows from the other properties of an ontological model for two very important classes of OPTs, namely, those that are \emph{deterministic} (i.e.~the set of probabilities is just $\{0,1\}$), and those where every probability distribution can be produced. Moreover, strongly causal OPTs where every probability distribution is a test are inherently \emph{convex}, in the sense that (as previously noted) convex combinations of tests can be understood as a specific form of conditioning based on probability distributions. Given that conditioning and probabilities are preserved by ontological models, convex combinations of tests, conditioned upon probability distributions, are consequently preserved as well, and hence ontological models preserve convex combinations of tests. The domain of ontological models can consist of either quotiented OPTs or unquotiented OPTs. We denote an ontological model of an OPT by $\xi$, or, when specifically considering the case of a quotiented OPT, we use $\tilde{\xi}$.

In the following sections, we will see that for ontological models of strongly causal quotiented OPTs, one can only consider the action of the map solely on transformations, regardless of the instrument that they belong to, and, as a by-product of this main result, this can be in turn used to demonstrate that ontological models of strongly causal quotiented OPTs where every probability distribution is a test are convex-linear when acting on transformations, which eventually can be extended to the general linear behaviour on transformations as shown in Appendix~\ref{appA}.

\subsection{Noncontextual ontological models}
Not all ontological models for an OPT $\Theta$ provide satisfying classical explanations for $\Theta$. In particular, scientific explanations are also expected to satisfy certain physical principles that have been successful elsewhere in physics. In particular, it has been argued that Leibniz's principle of the identity of indiscernibles~\cite{Leibniz} is an essential principle that should be followed in constructing candidate explanations, and it follows from this that an ontological model must satisfy the principle of generalized noncontextuality to be a satisfying explanation for a given operational theory. 

In brief, noncontextuality means that the representation of a given event in an ontological model should not depend on the \emph{context} of the event \cite{schmid2020unscrambling}. 
In the case of unquotiented OPTs, there are two relevant notions of context on which the representation $\xi(\text{t}|\text{T})$ of a given event, could depend. The first is variability within a given operational equivalence class of events; that is, for the representation to be noncontextual, $\xi(\text{t}|\text{T})$ should depend only on the equivalence class, so $\xi(\text{t}|\text{T})=\xi(\tilde{\text{t}}|\text{T})$. The second is variability of the test to which the event belongs, that is, for the representation to be noncontextual $\xi(\text{t}|\text{T})$ should not depend on $\text{T}$, that is, $\xi(\text{t}|\text{T})=\xi(\text{t})$.  Putting these two together noncontextual representations are those which satisfy $\xi(\text{t}|\text{T})=\xi(\tilde{\text{t}})$.

\begin{definition} \label{NCOMunqot}
(A noncontextual ontological model of a strongly causal unquotiented OPT). An ontological model of the strongly causal unquotiented OPT $\Theta$ $\input{Diagramss/49_ncontmodel.tikz}$ satisfies the principle of generalized noncontextuality if and only if it maps every member of an equivalence class of events to the same substochastic transformation, independent of the specific tests to which the events are associated, that is:
\begin{align}
\input{Diagramss/49_ncontmodel1.tikz}&\sim\input{Diagramss/49_ncontmodel2.tikz} \nonumber\\
&\Downarrow\nonumber \\
\input{Diagramss/49_ncontmodel3.tikz}&=\input{Diagramss/49_ncontmodel4.tikz}\;.
\end{align}
\end{definition}
This definition guarantees that the representation of events is independent of the specific test to which those events belong; namely, for every arbitrary test $\text{T}=\left[\text{t}_{x}\right]_{x\in X} $ in any unquotiented OPT, its representation in a noncontextual ontological model is given by
\begin{align}
\label{ncom}
\xi_{nc}\left(\text{T}\right) &= \left[\xi_{nc}\left(\text{t}_x\vert\text{T} \right)\right]_{x\in X}  \nonumber \\
&=  \left[\xi_{nc}\left( \text{t}_{x}\right)\right]_{x\in X}. 
\end{align}
In the following section it will be notationally convenient to work with fixed outcome sets $\mathbf{n}:=\{1,...,n\}$ and to denote $[\text{t}_i]_{i\in \mathbf{n}}$ as $[\text{t}_1,...,\text{t}_n]$. The results straightforwardly generalise to the case of more general outcome sets albeit at the expense of clarity of presentation. 

If we consider ontological models of quotiented OPTs, then the first kind of contexts (that is, variability within an operational equivalence class) is quotiented out. Hence, trivially, an ontological model of a quotiented OPT cannot depend on this kind of context. However, it is conceivable that ontological models of quotiented OPTs could still depend on the second kind of context, that is, to which instrument a given transformation belongs. That is, for every two arbitrary instruments  $\text{T}=\left[\text{t}^{*},\text{t}_{2},\ldots,\text{t}_{n}\right] $ and $\text{T}^{'}=[\text{t}^{*},\text{t}^{'}_{2},\ldots,\text{t}^{'}_{m}]$ from system A to system B in a given quotiented OPT, where the two have $\text{t}^{*}$  in common, it is conceivable that one could have
\begin{equation}
\label{crq}
\tilde{\xi}{\left(\text{t}^{*}\vert \text{T} \right)}\neq \tilde{\xi}( \text{t}^{*}\vert \text{T}^{'} ) .
\end{equation}
(This is not possible in the GPT framework, which is not equipped with instrument structure.)

In the following section, however, we will demonstrate that this is in fact not a possibility at all. That is, by utilizing some generic features of strongly causal quotiented OPTs, along with the properties we have defined for ontological models, we demonstrate that the instrument structure cannot be leveraged to allow for any contextual ontological representation of a strongly causal quotiented OPT.  Consequently, similar to in the GPT framework, there is a one-to-one correspondence between noncontextual ontological models of an unquotiented OPT and ontological models of the corresponding quotiented OPT.
\subsection{Impossibility of contextual ontological models of strongly causal quotiented OPTs }
To disprove the possibility of having contextual representation for ontological models of strongly causal quotiented OPTs, first we need to define a particular kind of coarse-graining. Given an arbitrary instrument  $\text{T}=\left[\text{t}_{1},\text{t}_{2},\ldots,\text{t}_{m}\right]$ from system A to system B in any strongly causal quotiented OPT, we define the coarse-graining $\mathfrak{C}_{\mathfrak{B}}{\left( \text{T}\right)}$ for instruments with outcome space cardinality greater than two as the one that results in the binary instrument
\begin{equation}
\label{BCG}
\mathfrak{C}_{\mathfrak{B}}{\left( \text{T}\right)}:=\left[ \text{t}_{1},\text{t}_{2}+\ldots +\text{t}_{m}\right].
\end{equation}
Note that any ontological model must commute with coarse-graining, which in this case means that $\tilde{\xi}\big(\mathfrak{C}_{\mathfrak{B}}(\text{T})\big)=\mathfrak{C}_{\mathfrak{B}}{\big(\tilde{\xi}(\text{T})\big)}$, so that
\begin{align}
\label{CG}
\left[\tilde{\xi}\big(\text{t}_{1}\vert\,\mathfrak{C}_{\mathfrak{B}}\left( \text{T}\right)\big),\tilde{\xi}\big(\text{t}_{2}+\ldots +\text{t}_{m}\vert\,\mathfrak{C}_{\mathfrak{B}}\left( \text{T}\right)\big)\right]\nonumber \\
= \left[\tilde{\xi}{\left(\text{t}_{1}\vert \text{T}\right)},\tilde{\xi}\left(\text{t}_{2}\vert \text{T}\right)+...+\tilde{\xi}\left(\text{t}_{m}\vert \text{T}\right)\right].
\end{align}
Having defined the above coarse-graining, consider two arbitrary instruments $\text{T}=[\text{t}^{*},\text{t}_{1},\ldots,\text{t}_{n}] $ and $\text{T}^{'}=[\text{t}^{*},\text{t}^{'}_{1},\ldots,\text{t}^{'}_{m}]$, both from system $A$ to system $B$, and with one transformation $\text{t}^{*}$ in common. Applying $\mathfrak{C}_{\mathfrak{B}}$ to these instruments, we have 
\begin{equation}
\mathfrak{C}_{\mathfrak{B}}\left( \text{T}\right)=[\text{t}^{*},\text{w}:=\sum_{i=1}^{n} \text{t}_{i}] , \; \mathfrak{C}_{\mathfrak{B}}( \text{T}^{'})=[\text{t}^{*},\text{w}^{'}:=\sum_{j=1}^{m} \text{t}^{'}_{j}] .
\end{equation}
For their ontological representations, using (\ref{CG}), we have
\begin{align}
\tilde{\xi}\big(\mathfrak{C}_{\mathfrak{B}}( \text{T})\big)&=\left[ \tilde{\xi}\big(\text{t}^{*}\vert \text{T}\big),\tilde{\xi}(\text{w}\,\vert\,\mathfrak{C}_{\mathfrak{B}}( \text{T})\big) \right] \\
\tilde{\xi}\big(\mathfrak{C}_{\mathfrak{B}}( \text{T}^{'})\big)&=\left[ \tilde{\xi} \big(\text{t}^{*}\,\vert\, \text{T}^{'}\,\big),\tilde{\xi}\big(\text{w}^{'}\vert\,\mathfrak{C}_{\mathfrak{B}}( \text{T}^{'})\big)\right].
\end{align}
Now, consider the transformations
\begin{equation}
\input{Diagramss/34.tikz}\;,\; \input{Diagramss/35.tikz}\;,
\end{equation}
namely the identity transformation  $\mathcal{I}_{\text{B}}$ on system B, and a second transformation, $\mathcal{D}_{\text{S}}$, that involves discarding system B and preparing it in an arbitrary deterministic state S of the theory. Since these transformations are deterministic, each can also be regarded as a singleton instrument, and we also use $\mathcal{I}_{\text{B}}$ and $\mathcal{D}_{\text{S}}$ to denote the singleton instruments. By implementing the instruments $\mathfrak{C}_{\mathfrak{B}}(\text{T})$ and $\mathfrak{C}_{\mathfrak{B}}(\text{T}^{'})$, and then applying $\mathcal{I}_{\text{B}}$ and $\mathcal{D}_{\text{S}}$, appropriately conditioned on their outcomes, we can define the following conditional pair of instruments from system A to system B:
\begin{equation}\label{eq:conditional-instruments}
\text{C}_{\text{T}}:=[\mathcal{I}_{\text{B}}\circ \text{t}^{*},\mathcal{D}_{\text{S}}\circ \text{w}],\;\text{C}_{\text{T}^{'}}:=[\mathcal{I}_{\text{B}}\circ \text{t}^{*},\mathcal{D}_{\text{S}}\circ \text{w}^{'}].
\end{equation}
Diagrammatically, these are
\begin{equation}\label{eq:conditioned-test-1}
\input{Diagramss/57_y.tikz}=[\input{Diagramss/57_y1.tikz}, \input{Diagramss/57_y2.tikz}]\;,
\end{equation}
\begin{equation}
\input{Diagramss/57_y3.tikz}=[ \input{Diagramss/57_y1.tikz}, \input{Diagramss/57_y4.tikz}]\;.
\end{equation}
By strong causality, both $\text{C}_{\text{T}}$ and $\text{C}_{\text{T}^{'}}$ are instruments within the theory. Moreover, by virtue of causality and the fact that the full coarse-graining of T and $\text{T}'$ are deterministic transformations, one has
\begin{equation}
\label{fullcg}
\mathcal{D}_{\text{S}}\circ \left( \text{t}^{*}+\text{w}\right)=\mathcal{D}_{\text{S}}=\mathcal{D}_{\text{S}}\circ ( \text{t}^{*}+\text{w}^{'})\Rightarrow\mathcal{D}_{\text{S}}\circ \text{w}=\mathcal{D}_{\text{S}}\circ \text{w}^{'},
\end{equation}
which implies that $\text{C}_{\text{T}}=\text{C}_{\text{T}^{'}}$. Now consider an ontological representation of $\text{C}_{\text{T}}$, which is an instrument from system $\Lambda_{\text{A}}$ to system $\Lambda_{\text{B}}$:
\begin{equation}
\label{CondOnt}
\tilde{\xi}{\left( \text{C}_{\text{T}}\right)}=\left[\tilde{\xi}{\left( \mathcal{I}_{\text{B}}\circ \text{t}^{*}\vert \text{C}_{\text{T}}\right)},\tilde{\xi}{\left(\mathcal{D}_{\text{S}}\circ \text{w}\vert \text{C}_{\text{T}}\right)}\right].
\end{equation}
Using conditioning preservation of ontological models, we have
\begin{equation}
\tilde{\xi}{\left( \text{C}_{\text{T}}\right)}=\left[ \tilde{\xi}{\left(\mathcal{I}_{\text{B}}\right)}\circ\tilde{\xi}{\left( \text{t}^{*}\vert \text{T}\right)},\tilde{\xi}{\left(\mathcal{D}_{\text{S}}\right)}\circ\tilde{\xi}\big( \text{w}\vert\,\mathfrak{C}_{\mathfrak{B}}{(\text{T})}\big)\right] ,
\end{equation} 
and using identity preservation, we further have
\begin{equation}\label{eq:final-instrument}
\tilde{\xi}{\left( \text{C}_{\text{T}}\right)}=\left[\tilde{\xi}{\left( \text{t}^{*}\vert \text{T}\right)},\tilde{\xi}{\left(\mathcal{D}_{\text{S}}\right)}\circ\tilde{\xi}\big( \text{w}\vert\,\mathfrak{C}_{\mathfrak{B}}(\text{T})\big)\right] .
\end{equation} 
It may help to repeat this logic diagrammatically. For the first transformation of (\ref{CondOnt}), one has
\begin{align}
\label{CondOnt1}
\tilde{\xi}{\left( \mathcal{I}_{\text{B}}\circ \text{t}^{*}\vert \text{C}_{\text{T}}\right)}=&\input{Diagramss/57_onty.tikz}\nonumber\\
=&\input{Diagramss/57_onty1.tikz}\nonumber\\
=&\input{Diagramss/57_onty3.tikz}\nonumber\\
=&\tilde{\xi}\big( \text{t}^{*}\vert\,\mathfrak{C}_{\mathfrak{B}}(\text{T})\big)=\tilde{\xi}\left( \text{t}^{*}\vert \text{T}\right);
\end{align}
for the second transformation, one has
\begin{align}
\tilde{\xi}\left(\mathcal{D}_{\text{S}}\circ \text{w}\vert \text{C}_{\text{T}}\right)=&\input{Diagramss/57_onty4.tikz}\nonumber\\
=&\input{Diagramss/57_onty5.tikz}\nonumber\\
=&\tilde{\xi}\left(\mathcal{D}_{\text{S}}\right)\circ\tilde{\xi}\big( \text{w}\vert\, \mathfrak{C}_{\mathfrak{B}}(\text{T}) \big).
\end{align}
Similarly, for $\text{C}_{\text{T}^{'}}$, we have
\begin{equation}
\tilde{\xi}( \text{C}_{\text{T}^{'}})=\left[ \tilde{\xi}( \text{t}^{*}\vert \text{T}^{'}),\tilde{\xi}\left(\mathcal{D}_{\text{S}}\right)\circ\tilde{\xi} \big( \text{w}^{'}\vert\, \mathfrak{C}_{\mathfrak{B}}(\text{T}^{'}) \big)\right] .
\end{equation}
But as $\text{C}_{\text{T}}=\text{C}_{\text{T}^{'}}$, we must have $\tilde{\xi}(\text{C}_{\text{T}})=\tilde{\xi}(\text{C}_{\text{T}^{'}})$ , which then immediately implies
\begin{equation}
\tilde{\xi}\left( \text{t}^{*}\vert \text{T}\right)=\tilde{\xi}( \text{t}^{*}\vert \text{T}^{'})
.
\end{equation}
This contradicts~\eqref{crq}, and so proves there is no possibility of a contextual ontological model of a quotiented OPT. It also shows that ontological representations of a quotiented OPT is determined entirely by its representation of the transformations, which naturally extend to give a unique representation of the instruments. Thus, one has the following:
\begin{lemma}
\label{Lemma}
For every instrument $\text{T}=\left[\text{t}_{1},\text{t}_{2},\ldots,\text{t}_{m}\right]$ in a strongly causal quotiented OPT, the corresponding instrument in an ontological model, from system $\Lambda_{\text{A}}$ to $\Lambda_{\text{B}}$, is represented as
\begin{align}
\tilde{\xi}\left( \text{T}\, \right) &=\left[\tilde{\xi}\left(\text{t}_{1}\vert \text{T} \,\right),\tilde{\xi}\left(\text{t}_{2}\vert \text{T} \,\right),\ldots,\tilde{\xi}\left(\text{t}_{m}\vert \text{T} \,\right)\right] \nonumber \\
&=\left[\tilde{\xi}(\text{t}_{1}),\tilde{\xi}(\text{t}_{2}),\ldots,\tilde{\xi}(\text{t}_{m})\right] .
\end{align}
\end{lemma}
Thus, despite having defined ontological models of quotiented OPTs based on instrument structure and exploring the potential to use this structure for contextual representations of quotiented OPTs, it turns out that due to some very generic properties of quotiented OPTs (particularly strong causality), one can demonstrate that it is sufficient to consider ontological models solely at the level of transformations, independent of the instruments they are associated with. Note that for ontological models of unquotiented OPTs, the same line of reasoning (with a minor modification for coarse-graining and replacing the sum with $\curlyvee$) cannot be applied. This is because, in unquotiented theories, (\ref{fullcg}) does not generally hold---the two sides of the equation are merely {\em operationally equivalent} rather than equal. Consequently, $\text{C}_\text{T}$ and $\text{C}_{\text{T}^{'}}$ are also operationally equivalent, but not equal, allowing for the possibility of an ontological model acting differently on them. 

However, ontological models of unquotiented theories that are noncontextual (Definition~\ref{NCOMunqot}) must satisfy the property in Lemma~\ref{Lemma}. Thus, one can follow a similar line of reasoning as in the proof of Proposition 3.2 from Ref.~\cite{Schmid2024structuretheorem} to obtain the following.
\begin{theorem}
There is a one-to-one correspondence between noncontextual ontological models of a strongly causal unquotiented OPT $\input{Diagramss/49_ncontmodel.tikz}$ and ontological models of the associated strongly causal quotiented OPT $\input{Diagramss/48_ontmodelq.tikz}$.
\end{theorem}

In particular, an ontological model of an unquotiented OPT is noncontextual if and only if it factors through the quotienting map. Note that if our result, Lemma~\ref{Lemma}, did {\em not} hold, then this would not be the case (and the above theorem would not hold), since then it would be possible that despite a given map (from the unquotiented OPT to the strictly classical OPT) factoring through the quotienting map, the ontological representation could still depend on the instrument structure (as this survives the quotienting map).

Before drawing our conclusions, we make here a few comments about the necessity of some of our assumptions. The proof of our main result does not technically rely on the preservation of parallel composition by the ontological model. However, we required that the ontological model is a strict monoidal functor---this being the hypothesis of \emph{diagram preservation} in Definition~\ref{ontmodel}---which includes the preservation of parallel composition:\footnote{Technically, our result would still hold under the weaker requirement that the ontological model is a \emph{strong} monoidal functor (we neglected this possibility here for the sake of simplicity).} this is because strict monoidality is a defining property of ontological models~\cite{Schmid2024structuretheorem}, and we therefore endorse it as reasonable in any case. On the other hand, note that diagram preservation is however necessary for proving linearity as done in Appendix~\ref{appA} (although linearity is not used for proving our main result).

Furthermore, preservation of the identity and of conditioning are manifestly important for the derivation of our main result, in particular to deduce Eq.~\eqref{eq:final-instrument} from Eq.~\eqref{CondOnt}. Moreover, there is evidence that, at least in some cases, full strong causality may be necessary for the existence of tests of the form~\eqref{eq:conditional-instruments} \emph{for every} transformation $\text{t}^*$ in an OPT. For instance, one can see that a broad class of \emph{minimal OPTs}~\cite{PhysRevA.109.022239,Rolino_2025} (violating strong causality)---such as minimal classical theory (MCT) studied in Ref.~\cite{PhysRevA.109.022239}, or minimal quantum theory (MQT) studied in Ref.~\cite{Rolino_2025}---lacks, for many transformations $\text{t}^{*}$, instruments of the form~\eqref{eq:conditional-instruments}. This strengthens the status of strong causality in the derivation of our main result.\footnote{For another link between strong causality, the geometry of transformations, and the test structure, see also~\cite[Thm.~2, App.~C]{d2020classicality}.}

\section{Discussion}
Our main result shows that for strongly causal quotiented OPTs, any ontological model cannot depend on the ``which instrument'' context. This means that once one has a specification of what the ontological model does to transformations, one can immediately say what happens to the instruments in the theory. This raises the question, why not forget about the structure of instruments entirely, and just define ontological models at the level of transformations?

This is essentially what was done in Ref.~\cite{Schmid2024structuretheorem}. The GPTs defined therein can be thought of as a strongly causal quotiented OPT where one does not have the instrument structure, and ontological models are defined therein as maps acting on the level of transformations. 

It is instructive to consider how these two frameworks match up with one another. On the one hand, we can turn an ontological model for a strongly causal quotiented OPT (where it is formally defined by its action on tests/instruments) into one for the GPT (where it is formally defined by its action on transformations). And, on the other hand, we can turn an ontological model for the GPT into one for the full strongly causal quotiented OPT. 
Importantly, one can show that going in either direction does indeed result in a valid ontological model. That is, the assumptions imposed in Ref.~\cite{Schmid2024structuretheorem} are in fact equivalent to those we demand here.
Hence, we have a one-to-one correspondence between noncontextual ontological models as defined in Ref.~\cite{Schmid2024structuretheorem} and those that we define here. At least for the study of noncontextual ontological models, the choice of which framework to work with is therefore simply a matter of personal preference.  

The most interesting direction for future work is to construct noncontextual ontological models for existing OPTs (or to prove that this is not possible). In forthcoming work, we begin this process by showing that one can construct a noncontextual ontological model for the theory known as bilocal classical theory \cite{d2020classicality}. Moreover, Theorem~2 in Ref.~\cite[App.~C]{d2020classicality} provides a procedure to construct a strongly causal completion for every quotiented causal OPT---namely, one with a unique deterministic effect (see Subsection~\ref{subsec:causality}). In other words, \emph{every quotiented causal OPT admits of a strongly causal (OPT-)explanation}. Therefore, a natural open problem arises: is it possible to combine our main result with the existence of strongly causal completions in order to conclude that contextual ontological models of quotiented causal OPTs are also impossible (thus extending our main result)? More generally, this work opens the door for relating topics studied in the framework of OPTs to generalized noncontextuality.

\

\section*{Acknowledgements}
We thank Rob Spekkens for helpful discussions on the possibility of contextual representations of quotiented theories. SS, ME, DS, and JHS were supported by the National Science Centre, Poland (Opus project, Categorical Foundations of
the Non-Classicality of Nature, project no. 2021/41/B/ST2/03149).
JHS conducted part of this research while visiting the
Okinawa Institute of Science and Technology (OIST) through the Theoretical Sciences Visiting
Program (TSVP). DS was supported by Perimeter Institute for Theoretical Physics. Research at Perimeter Institute is supported in part by the Government of Canada through the Department of Innovation, Science and Economic Development and by the Province of Ontario through the Ministry of Colleges and Universities. All of the diagrams within this manuscript were prepared using \href{https://tikzit.github.io/}{TikZit}.
\bibliography{context}
\appendix
\section{Empirical adequacy and linearity of ontological models}
\label{appA}

Since we have demonstrated in Lemma \ref{Lemma} that for strongly causal quotiented OPTs, the action of ontological models on transformations can be considered independently of the instrument to which they belong, it follows that for an arbitrary probability distribution $\text{P}_{\text{X}}=[p_{x}]_{x\in\text{X}}$, we have $\tilde{\xi}\left( p_{x}\vert\text{P}_{\text{X}}\right)=\tilde{\xi}\left(p_{x}\right)$. Then, every ontological model of a strongly causal quotiented OPT $\Theta$ satisfies the following properties:  
\begin{itemize}
    \item For every $p \in [0,1]$, it holds that $\tilde{\xi}(p) \in [0,1]$.
    \item The model preserves parallel compositions of probabilities, i.e., $\tilde{\xi}(pq) = \tilde{\xi}(p)\tilde{\xi}(q)$.
    \item The model preserves sums of probabilities, i.e., $\tilde{\xi}(p+q) = \tilde{\xi}(p) + \tilde{\xi}(q)$.
    \item Deterministic probabilities are mapped to deterministic probabilities, i.e., $\tilde{\xi}(1) = 1$.
\end{itemize}
The first property comes from the fact that an ontological model maps every probability distribution to a probability distribution. The second property follows from the property of diagram preservation, the third one from the preservation of coarse-graining, and the last one from the preservation of determinicity. Using the above properties, it is easy to prove that, at least if every probability distribution is a test of the OPT, then $\tilde{\xi}$ is linear on probabilities, as we now show. First---as it is proven in Appendix~\ref{app:prob-pres} just by using preservation of coarse-graining, of parallel composition, and of determinicity---one also has:
\begin{equation}\label{eq:emp-ad}
\tilde{\xi}(p)=p .
\end{equation}
That is, $\tilde{\xi}$ preserves probabilities. Property~\eqref{eq:emp-ad} has been termed \emph{empirical adequacy} in Ref.~\cite{Schmid2024structuretheorem}. This result can likewise be established for ontological models of strongly causal unquotiented OPTs. It is important to note, as mentioned earlier, that the proof outlined for Lemma~\ref{Lemma} is not applicable to ontological models of strongly causal unquotiented OPTs. Nevertheless, it remains valid for probability distributions. To see this, consider two probability distributions $\text{P}=[ p^{*},p_{1},p_{2},\ldots,p_{n}]$ and $\text{P}^{'}=[ p^{*},p^{'}_{1},p^{'}_{2},\ldots,p^{'}_{m}]$ in an strongly causal unquotiented OPT. Applying $\mathfrak{C}_{\mathfrak{B}}$ as defined in (\ref{BCG}), results in $\mathfrak{C}_{\mathfrak{B}}(\text{P})=[ p^{*},\sum_{i=1}^{n}p_{i}]$ and $\mathfrak{C}_{\mathfrak{B}}(\text{P}^{'})=[p^{*},\sum_{j=1}^{m}p^{'}_{j}]$ that are equal, and since any ontological model $\xi$ commutes with coarse-graining, it follows that
\begin{equation}
\xi(p^{*}\vert \text{P})=\xi(p^{*}\vert \text{P}^{'}).
\end{equation}  
Hence, test dependency can be dropped for probability distributions, namely, for every probability distribution  $\text{P}_{\text{X}}=[ p_{x}]_{x\in\text{X}}$ in any strongly causal unquotiented OPT, and for every ontological model $\xi$, it holds that $\xi{\left( p_{x}\vert\text{P}_{\text{X}}\right)}=\xi{\left(p_{x}\right)}$. Consequently, all the properties listed also apply to ontological models of unquotiented OPTs, ensuring the preservation of probability.
\begin{equation}
\xi(p)=p .
\end{equation}
For strongly causal quotiented OPTs where every probability distribution is a test, it can be demonstrated that $\tilde{\xi}$ is convex-linear with respect to transformations. Specifically, in any such OPT $\Theta$, for every probability distribution $[ p_{x} ]_{x \in \text{X}}$, every system A and B $\in \textbf{Sys}$, and for every transformation $\text{t}_{x} \in \text{Transf}(\Theta) \in \textbf{Transf}(\text{A} \rightarrow \text{B})$, one has
\begin{equation}
\label{convlin}
\tilde{\xi}\left(\sum_{x}p_{x}\text{t}_{x}\right)=\sum_{x} p_{x}\tilde{\xi}\left(\text{t}_{x}\right).
\end{equation}
This would indeed follow by applying Lemma~\ref{Lemma}, along with the preservation of conditioning, probabilities, and coarse-graining of $\tilde{\xi}$. To see this, consider two arbitrary instruments $\text{T}=[\text{t}_{1},\text{t}_{2}]$ and $\text{T}^{'}=[\text{t}^{'}_{1},\text{t}^{'}_{2}]$, from system A to system B, and a probability distribution $\text{P}=[ p, 1-p]$. Conditioning based on the probability distribution P, we can then define the conditional instrument
\begin{equation}
\text{C}_{\text{P}}:=\left[ p \text{t}_{1}, p \text{t}_{2}, (1-p)\text{t}^{'}_{1}, (1-p)\text{t}^{'}_{2}\right] .
\end{equation}
In other words T is implemented with the probability $p\in \text{P}$, and $\text{T}^{'}$ with the probability $1-p\in \text{P}$. Applying an ontological model $\tilde{\xi}$ on the above conditional instrument, and using conditioning preservation, probability preservation, and the fact the one can drop instrument dependency from $\tilde{\xi}$ in the light of Lemma~\ref{Lemma}, one has
\begin{equation}
\tilde{\xi}(\text{C}_{\text{P}})=\left[p\tilde{\xi}(\text{t}_{1}), p \tilde{\xi}(\text{t}_{2}), (1-p)\tilde{\xi}(\text{t}^{'}_{1}), (1-p)\tilde{\xi}(\text{t}^{'}_{2})\right].
\end{equation}
Now, given the outcome space $\text{X}:=\lbrace 1,2,3,4 \rbrace $ in the instrument $\text{C}_{\text{P}}$, consider the partition $\mathcal{K}(\text{X})=\lbrace\lbrace 1,3 \rbrace,\lbrace 2\rbrace, \lbrace 4\rbrace\rbrace$. Applying coarse-graining map based on this partition results in
\begin{equation}
\scalebox{0.9}{$
\mathfrak{C}_{\mathcal{K}(\text{X})}{\left(\tilde{\xi}(\text{C}_{\text{P}})\right)}=\left[p\tilde{\xi}(\text{t}_{1})+(1-p)\tilde{\xi}(\text{t}^{'}_{1}), p \tilde{\xi}(\text{t}_{2}), (1-p)\tilde{\xi}(\text{t}^{'}_{2})\right] .$}
\end{equation}
However, note that every ontological model, by definition, preserves coarse-graining, that is
\begin{equation*}
\begin{split}
\mathfrak{C}_{\mathcal{K}(\text{X})}&{\left(\tilde{\xi}(\text{C}_{\text{P}})\right)}=\tilde{\xi}\Big(\mathfrak{C}_{\mathcal{K}(\text{X})}(\text{C}_{\text{P}})\Big)\\
&=\left[\tilde{\xi}( p\text{t}_{1}+(1-p)\text{t}^{'}_{1}),  \tilde{\xi}(p\text{t}_{2}), \tilde{\xi}\big((1-p)\text{t}^{'}_{2}\big)\right] .
\end{split}
\end{equation*}
And hence, one has
\begin{equation}
\tilde{\xi}( p\text{t}_{1}+(1-p)\text{t}^{'}_{1})=p\tilde{\xi}(\text{t}_{1})+(1-p)\tilde{\xi}(\text{t}^{'}_{1}).
\end{equation}
Which can be easily generalized to (\ref{convlin}).\footnote{
It should be noted that, since it is unclear whether Lemma~\ref{Lemma} holds for strongly causal unquotiented OPTs, and because coarse-graining between events requires the $\curlyvee$ operation instead of summation, the conclusion for ontological models of such OPTs—based on the same reasoning—is that $\xi{\left( \, p\text{t}_{1}\curlyvee(1-p)\text{t}^{'}_{1} \,\, \big\vert \,\, \mathfrak{C}_{\mathcal{K}(\text{X})}(\text{C}_{\text{P}})\,\right)}=p\xi(\text{t}_{1}\vert \text{T})+(1-p)\xi(\text{t}^{'}_{1}\vert \text{T}^{'})$, where for noncontextual ontological models this reduces to $\xi_{nc}{\left( \, p\text{t}_{1}\curlyvee(1-p)\text{t}^{'}_{1}\right)}=p\xi_{nc}{(\text{t}_{1})}+(1-p)\xi_{nc}{(\text{t}^{'}_{1})}$.
}
Given that we have assumed the existence of the null transformation in quotiented OPTs, it follows from the preservation of coarse-graining that $\tilde{\xi}(\varepsilon_{\text{A} \rightarrow \text{B}})$, for every system A and B, is also a null transformation in any ontological model. Using this property of ontological models, along with the convex-linearity on transformations, it can be shown, by resorting to the same technique used in Ref.~\cite[App.~1]{hardy}, that $\tilde{\xi}$ is linear with respect to transformations, that is for every $r_{x}\in\mathbb{R}$, every system A and B $\in\textbf{Sys}(\Theta)$, and every $\text{t}_{x}\in\textbf{Transf}(\text{A}\rightarrow\text{B})$, such that $\sum_{x}r_{x}\text{t}_{x}\in\textbf{Transf}(\text{A}\rightarrow\text{B})$, one has
\begin{equation}
\tilde{\xi}{\left(\sum_{x}r_{x}\text{t}_{x}\right)}=\sum_{x} r_{x}\tilde{\xi}{\left(\text{t}_{x}\right)}.
\end{equation}

\section{Derivation of probability preservation}\label{app:prob-pres}

\begin{lemma}\label{lem:prob-pres}
    Let $f:[0,1]\to[0,1]$ be a multiplicative map from the real unit interval to itself:
    \begin{equation*}
        f{\left(pq\right)}=f{\left(p\right)}f{\left(q\right)},\quad \forall p,q\in[0,1];
    \end{equation*}
    then either $f$ is the zero map ($f=\mathtt{0}_{[0,1]}$)---i.e.~it maps everything to $0$---or it preserves the unit ($f{\left(1\right)}=1$). In the case where $f\neq\mathtt{0}_{[0,1]}$, suppose that $f$ is also additive:
    \begin{equation*}
    \begin{split}
    f{\left(p_1+p_2\right)}=f{\left(p_1\right)}+f{\left(p_2\right)},
    \\
    \forall p_1,p_2\in[0,1]:p_1+p_2\in[0,1]
    ;
    \end{split}
    \end{equation*}
    then $f$ is the identity map ($f=\mathtt{id}_{[0,1]}$)---i.e.~it maps everything to itself. In particular, any additive and multiplicative function on the real unit interval is either the zero map or the identity map.
\end{lemma}
    \begin{proof}
    To begin with, note that any multiplicative function $f:[0,1]\to[0,1]$ which does not preserve the unit must map everything to $0$, since
    \begin{equation*}
        f{\left(1\right)}f{\left(1\right)}=f{\left(1^2\right)}=f{\left(1\right)}\neq 1
    \end{equation*}
    implies $f{\left(1\right)}=0$, which in turn implies that
    \begin{equation*}
        f{\left(p\right)}=f{\left(1\cdot p\right)}=f{\left(1\right)}f{\left(p\right)}=0
    \end{equation*}
    for all $p\in[0,1]$. Conversely, let us now assume that $f$ preserves the unit, i.e.~$f{\left(1\right)}=1$ holds, and that $f$ is also additive on $[0,1]$. Then one has
    \begin{equation*}
        1=f{\left(1\right)}=f{\left(\frac{1}{2}+\frac{1}{2}\right)}=f{\left(\frac{1}{2}\right)}+f{\left(\frac{1}{2}\right)}
        ,
    \end{equation*}
    implying $f{\left(\frac{1}{2}\right)}=\frac{1}{2}$. Accordingly, for every \emph{dyadic rational}
    \begin{equation*}
        \frac{m}{2^n},\quad m,n\in\mathbb{N}
    \end{equation*}
    in the unit interval, one has also:
    \begin{equation*}
        f{\left(\frac{m}{2^n}\right)}=mf{\left(\frac{1}{2^n}\right)}=m\left[f{\left(\frac{1}{2}\right)}\right]^{n}=\frac{m}{2^n}
        .
    \end{equation*}
    Now, for all $p\in[0,1]$ such that $f{\left(p\right)}\neq p$, since $f{\left(1\right)}=1$, either $f{\left(p\right)}<p$ or $f{\left(1-p\right)}<1-p$ must hold. Let then $q$ denote either $p$ or $1-p$, and assume that $f{\left(q\right)}<q$. On the one hand, since the dyadic rationals are dense in the unit interval, there exist $a,b\in\mathbb{N}$ such that
    \begin{equation}\label{eq:mid-dyadic}
        f{\left(q\right)} < \frac{a}{2^b} < q
        .
    \end{equation}
    On the other hand, one thus gets:
    \begin{equation}\label{eq:majorisation}
        f{\left(q\right)}=f{\left(\frac{a}{2^b}+\epsilon\right)}=f{\left(\frac{a}{2^b}\right)}+f{\left(\epsilon\right)}=\frac{a}{2^b}+f{\left(\epsilon\right)}
        ,
    \end{equation}
    where the definition
    \begin{equation*}
        \frac{a}{2^b}+\epsilon\coloneqq q
    \end{equation*}
    with $\epsilon\in(0,1)$ has been introduced, which is well-posed by construction in the light of relation~\eqref{eq:mid-dyadic}. However, identity~\eqref{eq:majorisation} in turn means that
    \begin{equation*}
        f{\left(q\right)}\geq\frac{a}{2^b}
        ,
    \end{equation*}
    contradicting relation~\eqref{eq:mid-dyadic}. Therefore, one can finally conclude that $f{\left(p\right)}=p$ for all $p\in[0,1]$.
\end{proof}
In general, for an OPT there are two possibilities: either the OPT is \emph{deterministic}---namely, $\textbf{Transf}{(\mathrm{I}\to\mathrm{I})}=\{0,1\}$---or the OPT admits of at least one probability different from $0$ and $1$. For the ontological model of any deterministic OPT, preservation of coarse-graining and of determinicity straightforwardly imply empirical adequacy (probability preservation). On the other hand, in a nondeterministic OPT, \emph{every} probability distribution can be approximated with arbitrary precision---see e.g.~Exercise 3.3 of Ref.~\cite[pp.~122-126-127]{d2017quantum}. Accordingly, if there exists one probability different from $0$ and $1$, \emph{closure under operational limits}~\cite{chiribella2016quantum,d2017quantum,Rolino_2025} guarantees not only that $\textbf{Transf}{(\mathrm{I}\to\mathrm{I})}=[0,1]$, but also that \emph{every arbitrary probability distribution} is a valid instrument of the OPT: thereby, in this case Lemma~\ref{lem:prob-pres} suffices to ensure empirical adequacy also for the ontological models of any (generally nondeterministic) OPT.

\end{document}